%% file: ms.tex
\documentclass[11pt]{article} 

 \usepackage[margin=2.8cm]{geometry}
 
\usepackage[utf8x]{inputenc}
\usepackage{amsmath,amsfonts,amssymb,amsthm}
\usepackage{chngcntr}
\usepackage{cite}

\usepackage{enumitem}

\usepackage{tikz}
\usetikzlibrary{matrix,calc,positioning,fit,shapes,arrows,spy,backgrounds}
\usepackage{pgfplots}
\pgfplotsset{width=10cm,compat=newest}

\input{texMacros}
\input{tikzMacros}

\pgfplotsset{
  every axis plot/.append style={line width=0.8pt},
  every axis plot post/.append style={
    every mark/.append style={mark size=2.5}
  }
}

\newcounter{AppendixCount}

\title{Block- and Rank-Sparse Recovery for Direction Finding \\in Partly Calibrated Arrays}

\author{Christian Steffens and Marius Pesavento \\ 
		Communication Systems Group \\ 
		Darmstadt University of Technology, Germany \\
		e-mail: \{steffens, pesavento\}@nt.tu-darmstadt.de}
\date{February 2017}

\begin{document}


\maketitle


\begin{abstract}
\noindent
A sparse recovery approach for direction finding in partly calibrated arrays composed of subarrays with unknown displacements is introduced. The proposed method is based on mixed nuclear norm and $\ell_1$ norm minimization and exploits block-sparsity and low-rank structure in the signal model. For efficient implementation a compact equivalent problem reformulation is presented. The new technique is applicable to subarrays of arbitrary topologies and grid-based sampling of the subarray manifolds. In the special case of subarrays with a common baseline our new technique admits extension to a gridless implementation. As shown by simulations, our new block- and rank-sparse direction finding technique for partly calibrated arrays outperforms the state of the art method RARE in difficult scenarios of low sample numbers, low signal-to-noise ratio or correlated signals. 
\end{abstract}


\section{Introduction}
\noindent
Direction finding with sensor arrays has applications in various fields of signal processing such as wireless communications, radar, sonar or astronomy. In these types of applications it is desired to achieve a high angular resolution and to identify a large number of sources. This can be achieved by sensor arrays with a large aperture and a large number of sensors \cite{Krim:TwoDecades}. However, a large aperture size makes it difficult to achieve and maintain precise array calibration. Possible reasons for imperfect calibration are inaccuracies in the sensor positions, timing synchronization errors, or other unknown gain and phase offsets among sensors \cite{575894, 1165144, 543678, 340783, 509886, flanagan2001array, Gershman:Rare}. Standard approaches to this problem usually rely on either offline or online calibration. Offline calibration of the overall array is performed using reference sources at known positions and can easily become a challenging and time consuming task \cite{575894}. Alternatively, 
several online calibration techniques have been proposed which use calibration sources at unknown positions \cite{1165144, 543678, 340783, 509886, flanagan2001array}, but the complexity of these techniques is prohibitively high, and performance can be severely limited in the case of large sensor position errors \cite{flanagan2001array}. Moreover, these  techniques  cannot be employed in scenarios with imperfect time synchronization of sensors or other unknown sensor gain and phase offsets.

One way to overcome the calibration problem is to partition the overall array into smaller subarrays which are themselves comparably easy to calibrate. This type of array is referred to as partly calibrated array (PCA) and has received considerable interest in recent years. Generally, direction finding approaches for this type of arrays can be classified into non-coherent and coherent methods. In the non-coherent case the subarrays independently perform estimation of the directions of arrival (DoAs) or the signal covariance matrix to communicate these estimates to a central processor, where further processing is performed to achieve an improved joint estimate \cite{wax1985decentralized, stoica1995decentralized, 6854008, 1600024, 6880754}. In the coherent approach parameter estimation is performed based on joint coherent processing of all available sensor measurements, e.g., by computing a global sample covariance matrix, and imperfect calibration among the different subarrays is taken account of in the 
estimation process \cite{960398, Parvazi:PCA, 127959, 6811813, 6882325, Pesavento:Rare, Gershman:Rare, steffens2017shiftinvariance, 6882328}. In this work we consider the latter of the two approaches.


A prominent class of DoA estimation methods is based on subspace separation. In \cite{960398, Parvazi:PCA, 127959, 6811813, 6882325} the authors consider PCAs composed of multiple identical subarrays. Such types of array exhibit multiple shift invariances and methods such as the multiple invariant MUSIC and MODE \cite{960398} or multiple invariance ESPRIT \cite{127959, 6811813, 6882325} can be used to provide DoA estimates in a search-free fashion. In \cite{Pesavento:Rare} it is assumed that the PCA is composed of identically oriented linear arrays that can be transformed to a uniform linear array by linear translations of the subarrays. The authors present the root-RARE algorithm which admits search-free DoA estimation. The root-RARE method in \cite{Pesavento:Rare} was modified to the spectral RARE in \cite{Gershman:Rare} which admits application to arbitrary array topologies at the cost of increased computational complexity. Subspace-based methods are well investigated and are shown to asymptotically 
achieve an estimation performance close to the Cram\'er-Rao bound at low computational complexity. However, these subspace-based methods often have difficulties in certain practical scenarios. First, correlated source signals, e.g., in multipath environments, can significantly reduce the estimation performance. Second, subspace-based methods yield poor performance in the case of low number of snapshots and low signal-to-noise ratio, e.g., in fast changing environments. 

Recently, sparse recovery (SR) methods came into focus of DoA estimation studies. As reported in \cite{Malioutov:LassoDoa, steffens2016compact}, SR methods provide high-resolution parameter estimation performance without the aforementioned shortcomings of subspace-based methods. Moreover, SR methods are computationally tractable since they can be formulated as convex optimization problems. While classical SR methods aim at recovering sparse signal vectors and admit grid-based parameter estimation \cite{Tibshirani:Lasso, Chen98atomicdecomposition}, the special case of fully calibrated arrays (FCAs) of uniform linear topology with possibly missing sensors allows for gridless SR methods as proposed in \cite{candes2012a, candes2012b, tang2013}. In \cite{1600024, 6880754} the authors propose grid-based and gridless SR methods applicable for non-coherent processing in PCAs, where joint sparsity in the subarray signal representations is exploited. SR methods for coherent processing in PCAs have been presented in \
cite{steffens2017shiftinvariance, 6882328}. The method in \cite{steffens2017shiftinvariance} is based on the recently proposed SPARROW formulation \cite{steffens2016compact} and exploits multiple shift-invariances in PCAs composed of identical subarrays to provide gridless parameter estimation. In \cite{6882328} the well-known $\ell_{2,1}$ mixed-norm minimization approach \cite{Malioutov:LassoDoa} for FCAs is generalized to grid-based SR in PCAs of arbitrary topology by means of a mixed nuclear norm \cite{Boyd:RankMinimization, Recht2010} and $\ell_1$ norm, termed here as $\ell_{*,1}$ mixed-norm. As shown by numerical experiments \cite{6882328}, $\ell_{*,1}$ mixed-norm minimization clearly outperforms the spectral RARE \cite{Gershman:Rare} in frequency resolution performance for low signal-to-noise ratio and low number of snapshots.

In this paper we consider the $\ell_{*,1}$ mixed-norm minimization problem proposed in \cite{6882328} and derive an equivalent compact reformulation, termed as COmpact Block- and RAnk-Sparse recovery (\mname{}). The \mname{} formulation has a reduced number of optimization parameters as compared to the original $\ell_{*,1}$ mixed-norm minimization problem and we provide efficient implementations of the \mname{} formulation by means of semidefinite programming (SDP). While the SDP implementation is based on grid-based sampling of the subarray
manifolds and applicable to arbitrary array topologies, we furthermore present a search-free implementation of our \mname{} formulation for the special case of linear subarrays with a common baseline. We show by extensive numerical experiments that the \mname{} approach outperforms the state of the art methods in difficult scenarios. In summary, our main contributions are given as: 
\begin{itemize}
\item We introduce a sparse recovery approach for coherent processing in PCAs using $\ell_{*,1}$ mixed-norm minimization. 
\item We derive a compact reformulation of the $\ell_{*,1}$ mixed-norm minimization problem, termed as \mname{}. 
\item We develop a computationally efficient grid-based SDP implementation of the \mname{} formulation for arbitrary array topologies, and 
\item an efficient gridless SDP implementation of the \mname{} formulation for PCAs composed of subarrays with a common baseline. 
\end{itemize}

The paper is organized as follows: Section \ref{sec:sigModel} introduces the PCA signal model. The $\ell_{2,1}$ and $\ell_{*,1}$ mixed-norm minimization problems for FCAs and PCAs are discussed in Section \ref{sec:sota}. The \mname{} formulation is derived in Section \ref{sec:magrec} while grid-based and gridless SDP implementations are provided in Sections \ref{sec:Implementation} and \ref{sec:Gridless}. Numerical results are presented in Section \ref{sec:sims} before the paper is concluded in Section \ref{sec:conclusion}.

\textbf{Notation:} Boldface uppercase letters $\mb{X}$ denote matrices, boldface lowercase letters $\mb{x}$ denote column vectors, and regular letters $x,N$ denote scalars, with $\tj$ denoting the imaginary unit. Superscripts $\mb{X}^\tT$ and $\mb{X}^\tH$ denote transpose and conjugate transpose of a matrix $\mb{X}$, respectively. The term $\pdMat{P}{K}$ denotes the set of positive semidefinite block-diagonal matrices composed of $K$ blocks of size $P \times P$  on the main diagonal
. We write $[\mb{X}]_{m,n}$ to indicate the element in the $m$th row and $n$th column of matrix $\mb{X}$. The statistical expectation of a random variable $x$ is denoted as $\tE\{x\}$, and the trace of a matrix $\mb{X}$ is referred to as $\tr(\mb{X})$. The Frobenius norm and the $\ell_{p,q}$ mixed-norm of a matrix $\mb{X}$ are referred to as $\|\mb{X}\|_{\tF}$ and $\|\mb{X}\|_{p,q}$, respectively, 
while the $\ell_p$ norm of a vector $\mb{x}$ is denoted as $\|\mb{x}\|_p$. The term $\diag(x_1, \ldots, x_K)$
denotes a diagonal matrix with the elements $x_1, \ldots, x_K$ on its main diagonal while $\blkdiag(\mb{X}_1, \ldots, \mb{X}_K)$ denotes a block-diagonal matrix composed of submatrices $\mb{X}_1, \ldots, \mb{X}_K$ on its main block-diagonal.


\section{Signal Model} \label{sec:sigModel}
\noindent
Consider a linear array of arbitrary topology, composed of $M$ omnidirectional sensors, as depicted in Figure \ref{fig:PcaModel}. Assume the overall array is partitioned into $P$ subarrays with $M_p$ sensors in subarray $p$, for $p=1,\ldots,P$, such that $M=\sum_{p=1}^P M_p$. We define $\mb{\eta} = [ \eta^{(2)}, \ldots, \eta^{(P)} ]^\tT$ as the vector containing the $P-1$ unknown inter-subarray displacements $\eta^{(2)}, \ldots, \eta^{(K)}$ expressed in half signal wavelength and relative to the first subarray, i.e., $\eta^{(1)} = 0$. Furthermore, let $\rho_{m}^{(p)}$, for $m=1, \ldots, M_p$, $p=1,\ldots,P$, denote the perfectly known intra-subarray position of the $m$th sensor of subarray $p$ relative to the first sensor in the subarray, hence $\rho_{1}^{(p)} = 0$, and expressed in half signal wavelength. Consequently, the position of sensor $m$ in subarray $p$, relative to the first sensor in the first subarray, can be expressed as
\begin{align}
  r_{m}^{(p)} = \rho_{m}^{(p)} + \eta^{(p)} ,
  \label{eq:SenPos}
\end{align}
for $m=1,\ldots,M_p$ and $p=1,\ldots,P$. 

\begin{figure}[t!]
\begin{center}
  \small
  \input{modelImg1}
  \caption{Partly calibrated array composed of $M=9$ sensors partitioned in $P=3$ subarrays, and $L=2$ source signals}
  \label{fig:PcaModel}
  \vspace{-.4cm}
\end{center}  
\end{figure}
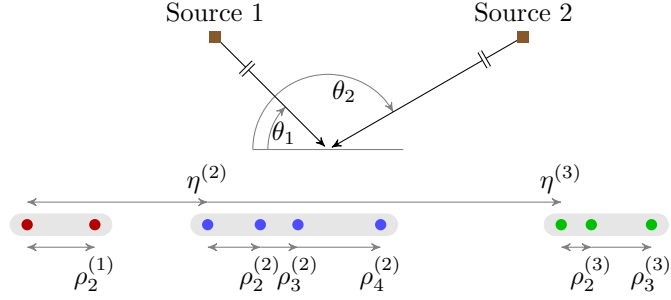

Moreover, assume a number of $L$ narrowband and far-field sources are illuminating the sensor array from angular directions $\theta_1, \ldots, \theta_L$, as illustrated in Figure \ref{fig:PcaModel}. The corresponding spatial frequencies are defined as $\mu_l = \cos \theta_l \in [-1,1)$, for $l = 1,\ldots,L$, and comprise the vector $\mb{\mu} = [\mu_1, \ldots, \mu_L]^\tT$. A total of $N$ signal snapshots are obtained at the output of each subarray $p$ and collected in the $M_p \times N$ subarray measurement matrix $\smash{\mb{Y}^{(p)}}$, for $p=1,\ldots,P$, where $\smash{[\mb{Y}^{(p)}]_{m,n}}$ denotes the output of the $m$th sensor in the $p$th subarray at time instant $n$. The subarray measurement matrices are collected in the $M \times N$ array measurement matrix $\mb{Y} = [\smash{\mb{Y}^{(1)\tT}}, \ldots, \smash{\mb{Y}^{(P)\tT}}]^\tT$, which is modeled as
\begin{equation}
 \mb{Y} = \mb{A}\left( \mb{\mu}, \mb{\eta} \right) \mb{\varPsi} + \mb{N}, 
 \label{eq:sigModelA}
\end{equation}
where $\mb{\varPsi} \in \mathbb{C}^{L \times N}$ is the source signal matrix and $\mb{N} \in \mathbb{C}^{M \times N}$ denotes a spatio-temporal white Gaussian sensor noise matrix. The $M \times L$ array steering matrix $\mb{A}\left( \mb{\mu}, \mb{\eta} \right)$ in \eqref{eq:sigModelA} is given by
\begin{align}
 \mb{A}(\mb{\mu},\mb{\eta}) = \left[ \mb{a} \left( \mu_1, \mb{\eta} \right), \ldots, \mb{a} \left( \mu_L, \mb{\eta} \right) \right], 
 \label{eq:SteerMatA}
\end{align}
and represents the response of the entire array, where $\mb{a}(\mu, \mb{\eta})$ denotes the steering vector for spatial frequency $\mu$ and subarray displacements $\mb{\eta}$. Based on the sensor position definition in \eqref{eq:SenPos}, the array steering vectors can be factorized as
\begin{align}
  \mb{a}(\mu, \mb{\eta}) 
  &= \mb{B}(\mu) \mb{\varphi}(\mu,\mb{\eta})
  \label{eq:a_mu}
\end{align}
where the $M \times P$ block-diagonal matrix 
\begin{align}
  \mb{B}(\mu)  = 
  \blkdiag \big( \mb{b}^{(1)}(\mu), \ldots, \mb{b}^{(P)}(\mu) \big)
  \label{eq:subarraySteeringMatrix}
\end{align}
contains the perfectly known subarray steering vectors
\begin{equation}
  \mb{b}^{(p)} \left( \mu \right) =  \big[ 1, e^{\tj \mu \, \rho_2^{(p)} }, \ldots, e^{\tj \mu \, \rho_{M_p}^{(p)} } \big]^\tT ,
  \label{eq:bk}
\end{equation}
for $p=1,\ldots,P$, on its diagonal, and the $L \times 1$ vector
\begin{align}
  \mb{\varphi}(\mu,\mb{\eta}) = [ 1, \alpha^{(2)} \te^{\tj\pi \mu \eta^{(2)}}, \ldots, \alpha^{(P)} \te^{\tj\pi \mu \eta^{(P)}} ]^\tT 
  \label{eq:Phik}
\end{align}
takes account of the subarray displacement shifts $\te^{\tj\pi \mu \eta^{(p)}}$, for $p=2,\ldots,P$, depending on the spatial frequencies in $\mb{\mu}$ and the subarray displacements in $\mb{\eta}$, and further unknown shifts $\alpha^{(p)}$, e.g., gain/phase or timing offsets among the subarrays \cite{Gershman:Rare}. In relation to \eqref{eq:SteerMatA}, let us define the $M \times PL$ matrix
\begin{align}
  \mb{B}(\mb{\mu}) = \big[ \mb{B}(\mu_1), \ldots, \mb{B}(\mu_L) \big]
  \label{eq:B}
\end{align}
containing all subarray responses for the spatial frequencies in $\mb{\mu}$, and the $PL \times L$ block-diagonal matrix
\begin{align}
  \mb{\varPhi}(\mb{\mu}, \mb{\eta}) &= 
  \blkdiag \big( \mb{\varphi}(\mu_1, \mb{\eta}),  \, \ldots, \, \mb{\varphi}(\mu_L, \, \mb{\eta}) \big) ,
  \label{eq:Phi}
\end{align}
composed of the subarray shift vectors in \eqref{eq:Phik}. Using \eqref{eq:B} and \eqref{eq:Phi}, the overall array steering matrix \eqref{eq:SteerMatA} can be factorized as
\begin{align}
  \mb{A}\left( \mb{\mu}, \mb{\eta} \right) =
  \mb{B}(\mb{\mu}) \; \mb{\varPhi}(\mb{\mu}, \mb{\eta} )
  \label{eq:ABPhi} 
\end{align}
such that the overall array measurement matrix in \eqref{eq:sigModelA} is equivalently modeled as
\begin{equation}
 \mb{Y} = \mb{B}(\mb{\mu}) \; \mb{\varPhi}(\mb{\mu}, \mb{\eta} ) \; \mb{\varPsi} + \mb{N},
 \label{eq:sigModelB}
\end{equation}
which forms the basis for the $\ell_{*,1}$ mixed-norm minimization problem discussed in the following section.


\section{State-of-the-Art} \label{sec:sota}
\noindent
In this section we will shortly review the $\ell_{2,1}$ mixed-norm minimization approach for FCAs before turning to the $\ell_{*,1}$ mixed-norm minimization approach for PCAs.

\subsection{Fully Calibrated Array}
\noindent
We first consider the case of an FCA where the subarray displacements in $\mb{\eta}$ are perfectly known. Based on the signal model in \eqref{eq:sigModelA} we introduce a sparse representation of the measurement matrix as
\begin{align}
  \mb{Y} = \mb{A}( \mb{\nu},\mb{\eta} ) \mb{X} + \mb{N}.
  \label{eq:SparseModelA}
\end{align}
The $M \times K$ overcomplete dictionary matrix $\mb{A}( \mb{\nu},\mb{\eta} )$ is obtained by sampling the field-of-view in $K \gg L$ spatial frequencies $\mb{\nu} = [ \nu_1, \ldots, \nu_K]^\tT$. For ease of presentation we assume that the frequency grid is sufficiently fine, such that the true frequencies in $\mb{\mu}$ are contained in the frequency grid $\mb{\nu}$, i.e., $\{\mu_l\}_{l=1}^L \subset \{\nu_k\}_{k=1}^K$. In Section \ref{sec:Gridless} we present an extension of our proposed formulation for subarrays with a common baseline which does not rely on the on-grid assumption. The $K \times N$ sparse signal matrix $\mb{X}$ in \eqref{eq:SparseModelA} contains elements
\begin{align}
  [\mb{X}]_{k,n} =&
  \begin{cases}
	[\mb{\varPsi}]_{l,n} \quad &\text{if } \nu_k = \mu_l \\
	0		  \quad &\text{else,}
  \end{cases}
  \label{eq:rowSparseStruct}
\end{align}
for $k=1,\ldots,K$, $l=1,\ldots,L$. Thus $\mb{X}$ exhibits a row-sparse structure, i.e., the elements in a row of $\mb{X}$ are either jointly zero or primarily non-zero. Based on the sparse representation \eqref{eq:SparseModelA}, the frequency estimation problem can be formulated as the mixed-norm minimization problem
\begin{equation}
  \min_{\mb{X}} 
  \frac{1}{2} \left\| \mb{A}(\mb{\nu}, \mb{\eta}) \; \mb{X} - \mb{Y} \right\|_\tF^2 + 
  \lambda \sqrt{N} \| \mb{X} \|_{2,0} ,
  \label{eq:mixedVectorNormL20}
\end{equation}
where $\lambda > 0$ is a regularization parameter determining the sparsity, i.e., the number of non-zero rows in the minimizer $\smash{\hat{\mb{X}}}$. Row-sparsity is enforced by minimizing the $\ell_{2,0}$ mixed-norm in \eqref{eq:mixedVectorNormL20}, which is defined as the number of non-zero rows $\mb{x}_k$ of the matrix $\mb{X}  = [ \mb{x}_1, \ldots, \mb{x}_K]^\tT$, i.e., the cardinality of the union support set according to
\begin{align}
  \| \mb{X} \|_{2,0} = \big| \{ k \,:\, \| \mb{x}_k \|_2 > 0 \} \big| .
\end{align}
Since the problem in \eqref{eq:mixedVectorNormL20} is NP-hard, several approximation methods have been proposed in the literature, including convex relaxation to the well-known $\ell_{2,1}$ mixed-norm minimization problem \cite{yuan2006grouplasso, kowalski2009mixednorm, Malioutov:LassoDoa, steffens2016compact}
\begin{equation}
  \min_{\mb{X}} 
  \frac{1}{2} \left\| \mb{A}(\mb{\nu}, \mb{\eta}) \; \mb{X} - \mb{Y} \right\|_\tF^2 + 
  \lambda \sqrt{N} \| \mb{X} \|_{2,1} .
  \label{eq:mixedVectorNorm}
\end{equation}
The $\ell_{2,1}$ mixed-norm in \eqref{eq:mixedVectorNorm} is defined as
\begin{align}
  \| \mb{X} \|_{2,1} = \sum_{k=1}^{K} \left\| \mb{x}_k \right\|_2
  \label{eq:lpqNorm}
\end{align}
and induces a non-linear coupling among the elements in each row $\mb{x}_k$, $k=1,\ldots,K$, of the matrix $\mb{X}$ such that the $\ell_1$ norm, i.e., the nonnegative summation, is performed on the $\ell_2$ norms of the rows in $\opt{\mb{X}}$. Given a minimizer $\smash{\hat{\mb{X}} = [ \hat{\mb{x}}_1, \ldots, \hat{\mb{x}}_K]^\tT}$ of \eqref{eq:mixedVectorNorm}, the frequency estimation problem reduces to finding the local maxima in the vector of the signal $\ell_2$ row-norms $ \hat{\mb{x}}^{\ell_2} = [ \| \hat{\mb{x}}_1 \|_2, \ldots, \| \hat{\mb{x}}_K \|_2]^\tT$ and assigning the corresponding frequency grid points to the set $\{ \hat{\mu} \}$ of estimated frequencies.


For the PCA case with uncertain array response $\mb{A}(\mb{\nu}, \mb{\eta})$ due to the unknown displacements in $\mb{\eta}$, the $\ell_{2,1}$ mixed-norm minimization approach in \eqref{eq:mixedVectorNorm} cannot be applied and a more sophisticated approach has to be devised, as discussed in the following subsection.

\subsection{Partly Calibrated Array}
\noindent
Analogous to the FCA case in \eqref{eq:SparseModelA}, we introduce a sparse representation of the signal model in \eqref{eq:sigModelB} for the PCA case as
\begin{equation}
 \mb{Y} = \mb{B}(\mb{\nu}) \; \mb{\varPhi}(\mb{\nu}, \mb{\eta} ) \; \mb{X} + \mb{N} ,
 \label{eq:SparseModelB}
\end{equation}
where the row-sparse matrix $\mb{X}$ is defined similarly as for the FCA case in \eqref{eq:rowSparseStruct}. Furthermore, the $M \times PK$ overcomplete subarray dictionary matrix $\mb{B}(\mb{\nu})$ and the $PK \times K$ overcomplete subarray shift matrix $\mb{\varPhi}(\mb{\nu}, \mb{\eta})$ are defined in correspondence to \eqref{eq:B} and \eqref{eq:Phi}, respectively.

In the PCA case, the inter-subarray displacements in $\mb{\eta}$ are unknown and thus represent additional estimation variables, hence the subarray shifts in $\mb{\varPhi}(\mb{\nu}, \mb{\eta} )$, which depend on the spatial frequencies in $\mb{\nu}$ and the subarray displacements $\mb{\eta}$, have to be appropriately included in the sparse estimation problem. To this end we introduce a model that couples among the variables $\mb{x}_k$ in the rows of $\mb{X} = [\mb{x}_1, \ldots, \mb{x}_K]^\tT$ and the subarray shifts in $\mb{\varphi}( \nu_k, \mb{\eta} )$, for $k=1,\ldots,K$. We define the $KP \times N$ extended signal matrix $\mb{Q}$ as
\begin{align}
  \mb{Q} &= \mb{\varPhi}(\mb{\nu}, \mb{\eta} ) \; \mb{X} \label{eq:Q} 
\end{align}
containing the products of the subarray shifts and the signal waveforms. Note that in this formulation the number of the unknown complex-valued signal variables is increased to $K P N$ elements in the matrix $\mb{Q}$, as compared to the total $K (N + P -1)$ complex-valued unknowns in both $\mb{X}$ and $\mb{\varPhi}(\mb{\nu},\mb{\eta})$. On the other hand, due to the block structure of the subarray shift matrix $\mb{\varPhi}(\mb{\nu},\mb{\eta})$ as defined in \eqref{eq:Phi}, the matrix $\mb{Q} = [\mb{Q}_1^\tT, \ldots, \mb{Q}_K^\tT]^\tT$ in \eqref{eq:Q} enjoys a special structure as it is composed of $K$ stacked rank-one matrices 
\begin{equation}
  \mb{Q}_k = \mb{\varphi}(\nu_k,\mb{\eta}) \, \mb{x}^\tT_k, \quad \text{ for } k=1,\ldots,K . 
  \label{eq:Qn}
\end{equation}
Using the formulation in \eqref{eq:Q}, the sparse representation for the PCA case in \eqref{eq:SparseModelB} is equivalently described by 
\begin{align}
  \mb{Y} = \mb{B} \mb{Q} + \mb{N} ,
  \label{eq:sigModelQ}
\end{align}
where for ease of presentation we use $\mb{B} = \mb{B}(\mb{\nu})$ to denote the dictionary matrix in \eqref{eq:sigModelQ} and throughout the paper. An SR approach to take account of the special structure of the signal matrix $\mb{Q}$ in \eqref{eq:Q} is given as
\begin{equation}
  \min_{\mb{Q}} 
  \frac{1}{2} \left\| \mb{B} \mb{Q} - \mb{Y} \right\|_\tF^2 + 
  \lambda \sqrt{N} \sum_{k=1}^K \rank (\mb{Q}_k) .
  \label{eq:RankProblem}
\end{equation} 
The formulation in \eqref{eq:RankProblem} takes twofold advantage of the sparsity assumption. First, minimization of the rank-terms encourages low-rank blocks $\hat{\mb{Q}}_1, \ldots, \hat{\mb{Q}}_K$ in the minimizer $\hat{\mb{Q}}$. Second, minimizing the sum-of-ranks provides a block-sparse structure of $\hat{\mb{Q}}$, i.e., the elements in each block $\hat{\mb{Q}}_k$, for $k=1,\ldots,K$, are either jointly zero or primarily non-zero. However, the problem in \eqref{eq:RankProblem} is NP-hard and computationally intractable. 

The nuclear norm represents a tight convex approximation of the rank function and it has been successfully applied in a variety of rank minimization problems \cite{Boyd:RankMinimization, Recht2010}. The definition of the nuclear norm is given as
\begin{equation}
    \left\| \mb{Q}_k \right\|_{*} = \tr \big( (\mb{Q}_k^\tH \mb{Q}_k)^{1/2} \big) = \sum_{i=1}^{r} \sigma_{k,i},
\end{equation}
where $r = \min(P,N)$ and $\sigma_{k,i}$ is the $i$th singular value of $\mb{Q}_k$. 
Along these lines it has been proposed in \cite{6882328} to approximate the sparse estimation problem \eqref{eq:RankProblem} by the following convex minimization problem
\begin{equation}
  \min_{\mb{Q}} 
  \frac{1}{2} \left\| \mb{B} \mb{Q} - \mb{Y} \right\|_\tF^2 + 
  \lambda \sqrt{N} \left\| \mb{Q} \right\|_{*,1},
  \label{eq:ConvProblem}
\end{equation} 
where $\left\| \mb{Q} \right\|_{*,1}$ denotes the $\ell_{*,1}$ mixed-norm, computed as
\begin{align}
  \left\| \mb{Q} \right\|_{*,1} = \sum_{k=1}^{K} \left\| \mb{Q}_k \right\|_{*} .
  \label{eq:Nuc1MixedNorm}
\end{align}
Similar to \eqref{eq:RankProblem}, the problem in \eqref{eq:ConvProblem} motivates low-rank blocks $\hat{\mb{Q}}_1, \ldots, \hat{\mb{Q}}_K$ and a block-sparse structure in the minimizer $\hat{\mb{Q}} = [ \hat{\mb{Q}}_1^\tT, \ldots, \hat{\mb{Q}}_K^\tT ]^\tT$. Note that the PCA formulations in \eqref{eq:RankProblem} and \eqref{eq:ConvProblem} reduce to the FCA formulations in \eqref{eq:mixedVectorNormL20} and \eqref{eq:mixedVectorNorm}, respectively, in the case of a single subarray, i.e., $P=1$.

Performing singular value decomposition on the matrix blocks in $\hat{\mb{Q}}$, i.e.,
\begin{align}
  \hat{\mb{Q}}_k = \hat{\mb{U}}_k \hat{\mb{\Sigma}}_k \hat{\mb{V}}_k^\tT \quad \text{for } k=1,\ldots,K,
\end{align}
the signal waveform $\hat{\mb{x}}_k$ and subarray shifts $\hat{\mb{\varphi}}(\nu_k,\mb{\eta})$ corresponding to the spatial frequency $\nu_k$ can be recovered according to
\begin{align}
 \hat{\mb{x}}_k = \hat{\sigma}_{k,1} \, [\hat{\mb{u}}_{k,1}]_1 \, \hat{\mb{v}}_{k,1}
  \quad \text{and} \quad
 \hat{\mb{\varphi}}(\nu_k,\mb{\eta}) = \frac{\hat{\mb{u}}_{k,1}}{[\hat{\mb{u}}_{k,1}]_1} .
 \label{eq:sigPhase}
\end{align}
The left and right singular vectors $\hat{\mb{u}}_{k,1}$ and $\hat{\mb{v}}_{k,1}$ in \eqref{eq:sigPhase} correspond to the largest singular value $\hat{\sigma}_{k,1}$ of $\hat{\mb{Q}}_k$ and normalization to the first element $[\hat{\mb{u}}_{k,1}]_1$ of $\hat{\mb{u}}_{k,1}$ in \eqref{eq:sigPhase} is performed to take account of the structure of the subarray shift vectors $\hat{\mb{\varphi}}(\nu_k,\mb{\eta})$ according to \eqref{eq:Phik}.

While the $\ell_{*,1}$ mixed-norm minimization problem in \eqref{eq:ConvProblem} provides a tractable approach for SR in PCAs, it suffers from high computational complexity in the case of a large number of snapshots $N$ and grid points $K$. To overcome this difficulty we provide in the following section a compact reformulation of problem \eqref{eq:ConvProblem}.


\section{Compact Block- and Rank-Sparse Recovery} \label{sec:magrec}
\noindent
One of the main results of this paper is formulated in the following theorem:
\begin{theorem}[Problem Equivalence] \label{th:equivalence}
The block- and rank-sparsity inducing $\ell_{*,1}$ mixed-norm minimization problem 
\begin{equation}
  \min_{\mb{Q}} 
  \frac{1}{2} \left\| \mb{B} \mb{Q} - \mb{Y} \right\|_\tF^2 + 
  \lambda \sqrt{N} \left\| \mb{Q} \right\|_{*,1} 
  \label{eq:mixedVectorNorm_v2}
\end{equation}
is equivalent to the convex problem 
\begin{align}
  \min_{\mb{S} \in \pdMat{P}{K}} \tr \big( (\mb{B} \mb{S} \mb{B}^\tH  + \lambda \mb{I})^{-1} \hat{\mb{R}} \big) + \tr(\mb{S}), \label{eq:smr1}  
\end{align}
with $\hat{\mb{R}} = \mb{Y} \mb{Y}^\tH /N$ and $\pdMat{P}{K} $ denoting the sample covariance matrix and the set of positive semidefinite block-diagonal matrices composed of $K$ blocks of size $P \times P$, respectively. The equivalence holds in the sense that a minimizer  $\opt{\mb{Q}}$ for problem \eqref{eq:mixedVectorNorm_v2} can be factorized as
\begin{align}
  \opt{\mb{Q}} = \opt{\mb{S}} \mb{B}^\tH  ( \mb{B} \opt{\mb{S}} \mb{B}^\tH + \lambda \mb{I} )^{-1} \mb{Y} 
  \label{eq:smr2}
\end{align}
where $\opt{\mb{S}}$ is a minimizer for problem \eqref{eq:smr1}.
\end{theorem}

A proof of the equivalence is provided in Appendix \ref{sec:proof}, while a proof of the convexity of \eqref{eq:smr1} is provided in Section \ref{sec:SDP} by establishing equivalence to a semidefinite program. 

In addition to relation \eqref{eq:smr2}, it can be shown (see Appendix \ref{sec:proof}) that a minimizer $\opt{\mb{S}}=\blkdiag(\opt{\mb{S}}_1, \ldots, \opt{\mb{S}}_K)$ of \eqref{eq:smr1} relates to the signal matrix $\opt{\mb{Q}} = [\opt{\mb{Q}}_1^\tT, \ldots, \opt{\mb{Q}}_K^\tT]^\tT$ according to
\begin{align}
  \opt{\mb{S}}_k = \frac{1}{\sqrt{N}} ( \opt{\mb{Q}}_k \opt{\mb{Q}}_k^\tH )^{1/2},
  \label{eq:magIdentity}
\end{align}
for $k=1,\ldots,K$, such that the block-support of $\opt{\mb{Q}}$ is equivalently represented by the block-support of the matrix $[\opt{\mb{S}}_1^\tT, \ldots, \opt{\mb{S}}_K^\tT]^\tT$. Similarly, the rank of the matrix blocks $\opt{\mb{Q}}_k$ is equivalently represented by the matrix blocks $\opt{\mb{S}}_k$, i.e., $\rank(\opt{\mb{Q}}_k) = \rank(\opt{\mb{S}}_k)$, for $k=1,\ldots,K$. 

We observe that the problem in \eqref{eq:smr1} only relies on the measurement matrix $\mb{Y}$ through the sample covariance matrix $\hat{\mb{R}}$, leading to a significantly reduced problem size,  especially  in  the  case  of  large  number  of  snapshots $N$. In this context we term the formulation in \eqref{eq:smr1} as COmpact Block- and RAnk-Sparse recovery (\mname{}). The compact formulation \eqref{eq:smr1} contains $K P^2$ real-valued optimization parameters in the positive semidefinite matrix $\mb{S}$, as opposed to the $2KPN$ real-valued optimization parameters in $\mb{Q}$ in problem \eqref{eq:mixedVectorNorm_v2}. Consequently, in the case of a large number of snapshots $N > P/2$, the reformulation \eqref{eq:smr1} has reduced computational complexity as compared to \eqref{eq:mixedVectorNorm_v2}. 


\section{Implementation of the \mname{} Formulation} \label{sec:Implementation}
\noindent
The $\ell_{*,1}$ mixed-norm minimization problem \eqref{eq:mixedVectorNorm_v2} has been well investigated in literature and implementations based on the coordinate descent method \cite{6882328}, the STELA algorithm \cite{steffens2016mimo} and semidefinite programming (SDP) \cite{steffens2016noncircular} have been proposed. Here we will shortly revise the SDP implementation of the $\ell_{*,1}$ mixed-norm minimization problem \eqref{eq:mixedVectorNorm_v2} to highlight the reduction in computational complexity obtained by employing the \mname{} formulation in \eqref{eq:smr1}. 

\subsection{SDP Form of the $\ell_{*,1}$ Mixed-Norm Minimization Problem} \label{sec:MixedNormImplementation}
\noindent 
As discussed in \cite{Boyd:RankMinimization}, minimization of the nuclear norm 
\begin{align}  
	\min_{\mb{Q}_k \in \mathcal{C}} \; \left\| \mb{Q}_k \right\|_{*},
	\label{eq:NucNormSdp0}
\end{align}
for some convex set $\mathcal{C}$, can be expressed as the SDP
\begin{subequations}
\label{eq:NucNormSdp1}
\begin{align}  
    \min_{\substack{\mb{Q}_k \in \mathcal{C}, \\ \mb{P}_{k,1}, \mb{P}_{k,2}}} 	& \;\; \frac{1}{2} \big( \tr (\mb{P}_{k,1}) + \tr (\mb{P}_{k,2}) \big) \\
	\tst 	& \; \mtx{ \mb{P}_{k,1} & \mb{Q}_k \\ \mb{Q}_k^\tH & \mb{P}_{k,2} } \succeq \mb{0} ,	
\end{align}
\end{subequations}
where $\mb{P}_{k,1} = \mb{P}_{k,1}^\tH$ and $\mb{P}_{k,2} = \mb{P}_{k,2}^\tH$ are auxiliary variables of size $P \times P$ and $N \times N$, respectively. The SDP formulation \eqref{eq:NucNormSdp1} admits simple implementation of the nuclear norm minimization problem using standard convex solvers, such as SeDuMi \cite{S98guide}.

Based on the equivalence of \eqref{eq:NucNormSdp0} and \eqref{eq:NucNormSdp1}, the $\ell_{*,1}$ mixed-norm minimization problem \eqref{eq:mixedVectorNorm_v2} can be equivalently formulated as 
\begin{subequations}
\label{eq:NucNormSdp2}
\begin{align}
  \smash{ \min_{\substack{\{\mb{Q}_k, \mb{P}_{k,1}, \\ \mb{P}_{k,2}\}}} } & \; 
  \left\| \mb{B} \mb{Q} - \mb{Y} \right\|_\tF^2 + \lambda \sqrt{N} \sum_{k=1}^{K} \tr (\mb{P}_{k,1}) + \tr (\mb{P}_{k,2})  \\
  \tst 	& \mtx{ \mb{P}_{k,1} & \mb{Q}_k \\ \mb{Q}_k^\tH & \mb{P}_{k,2} } \succeq 0, \text{ for } k=1,\ldots,K .
\end{align}
\end{subequations}
Note that with the auxiliary variables in $\mb{P}_{k,1}$ and $\mb{P}_{k,2}$, for $k=1,\ldots,K$, the problem \eqref{eq:NucNormSdp2} has $K (P+N)^2$ real-valued optimization variables as opposed to the problem formulation in \eqref{eq:mixedVectorNorm_v2} which has $2KPN$ real-valued optimization variables. For a large number of grid points $K$ or snapshots $N$ the SDP formulation \eqref{eq:NucNormSdp2} becomes intractable and alternative implementations are required as presented in the next subsection. We remark that the problem of large snapshot number has been addressed in previous literature by matching the signal subspace of the measurements $\mb{Y}$ instead of the measurements itself, see \cite{Malioutov:LassoDoa, steffens2016noncircular}, leading to a reduced number of effective signal snapshots, however at the expense of potential performance degradation, e.g., in the case of correlated source signals. 

\def\twa{3cm}
\def\twb{3.8cm}
\begin{table}[t]
\begin{center}
\small
\begin{tikzpicture}
\matrix (first) [table]
{
\node[draw=white, fill=white] {}; 
  &[2pt]  \parbox{\twb}{$\ell_{*,1}$ Mixed-Norm \eqref{eq:NucNormSdp2}} 
  & \parbox{\twb}{\mname{} \eqref{eq:sdp1}} 
  & \parbox{\twb}{\mname{} \eqref{eq:sdp1b}} \\
\parbox{\twa}{Number of real\\ parameters}	
	& \parbox{\twb}{ $K(P+N)^2$} 
	& \parbox{\twb}{ $KP^2 + N^2$}
	& \parbox{\twb}{ $KP^2 + M^2 $} \\
\parbox{\twa}{Number $\times$ size \\of SDP constraints}
	& \parbox{\twb}{$K \times \{ (P \!+\! N) \!\times\! (P \!+\! N) \}$}	
	& \parbox{\twb}{$K \times \{ P \!\times\! P \}$ and \\ 
					$1 \times \{ (M \!+\! N) \!\times\! (M \!+\! N) \}$}
	& \parbox{\twb}{$K \times \{ P \!\times\! P \}$ and \\ 
					$1 \times \{ (2M) \!\times\! (2M) \}$} \\
};
\end{tikzpicture}
\vspace{-.3cm}
\caption{Comparison of Equivalent SDP Implementations}
\label{tab:sdp}

\end{center}
\end{table}

\subsection{SDP Form of the \mname{} Method} \label{sec:SDP}
\noindent 
In order to solve the \mname{} formulation in \eqref{eq:smr1} by means of a tractable SDP which can be treated by standard convex solvers consider the following corollaries \cite{vandenberghe1996semidefinite}:
\begin{corollary} \label{col:sdp1}
The \mname{} formulation in \eqref{eq:smr1} is equivalent to the convex semidefinite program
\begin{subequations}
\label{eq:sdp1}
\begin{align}
  \min_{\mb{S}, \mb{Z}_N} \; & \;\; \frac{1}{N} \tr( \mb{Z}_N ) + \tr( \mb{S} ) \\
  {\rm s.t. } & \; \mtx{ \mb{Z}_N & \mb{Y}^\tH \\ \mb{Y} & 
  \mb{B} \mb{S} \mb{B}^\tH + \lambda \mb{I} } \succeq \mb{0}  \label{eq:sdp1Con1} \\
  & \; \; \mb{S} \in \pdMat{P}{K} ,
\end{align}	
\end{subequations}
where $\mb{Z}_N$ is a Hermitian matrix of size $N \times N$.
\end{corollary}
To see the equivalence between the two problems we note that $\mb{B} \mb{S} \mb{B}^\tH + \lambda \mb{I} \succ \mb{0}$ is positive definite for any $\lambda > 0$ and consider the Schur complement of the constraint \eqref{eq:sdp1Con1}
\begin{align}
  \mb{Z}_N \succeq \mb{Y}^\tH (\mb{B} \mb{S} \mb{B}^\tH + \lambda \mb{I})^{-1} \mb{Y}
  \label{eq:sdpEquivalence1}
\end{align}
which implies
\begin{align}
  \frac{1}{N} \tr(\mb{Z}_N) &\geq 
  \frac{1}{N} \tr \big( \mb{Y}^\tH (\mb{B} \mb{S} \mb{B}^\tH + \lambda \mb{I})^{-1} \mb{Y} \big) \nonumber \\
						  &= \tr \big( (\mb{B} \mb{S} \mb{B}^\tH + \lambda \mb{I}_M)^{-1} \hat{\mb{R}} \big) .
  \label{eq:sdpEquivalence2}
\end{align}
Since in problem \eqref{eq:sdp1} $\tr(\mb{Z}_N)$ is minimized, it can be proved by contradiction that the relation in \eqref{eq:sdpEquivalence2} must hold with equality, proving the equivalence of \eqref{eq:smr1} and \eqref{eq:sdp1}. 
\begin{corollary} \label{col:sdp2}
The \mname{} formulation in \eqref{eq:smr1} admits the equivalent problem formulation
\begin{subequations}
\label{eq:sdp1b}
\begin{align}
  \min_{\mb{S}, \mb{Z}_M} \; & \;\; \tr( \mb{Z}_M \hat{\mb{R}} ) + \tr( \mb{S} ) \\
  {\rm s.t. } & \; \mtx{ \mb{Z}_M & \mb{I}_M \\ \mb{I}_M & \mb{B} \mb{S} \mb{B}^\tH + \lambda \mb{I}_M } \succeq \mb{0} \label{eq:sdp1bCon1} \\
  & \; \; \mb{S} \in \pdMat{P}{K}
\end{align}	
where $\mb{Z}_M$ is a Hermitian matrix of size $M \times M$.
\end{subequations}
\end{corollary}
The proof to Corollary \ref{col:sdp2} follows similar arguments as in the proof of Corollary \ref{col:sdp1} and is therefore omitted here. In contrast to \eqref{eq:sdp1}, the size of the semidefinite constraint in \eqref{eq:sdp1b} is independent of the number of snapshots $N$. It follows that either problem formulation \eqref{eq:sdp1} or \eqref{eq:sdp1b} can be selected to solve \eqref{eq:smr1}, depending on the number of snapshots $N$ and the resulting size of the semidefinite constraint. The problems \eqref{eq:sdp1} and \eqref{eq:sdp1b} have $KP^2$ real-valued optimization variables in $\mb{S}$ and additional $N^2$ or $M^2$ real-valued parameters in $\mb{Z}_N$ and $\mb{Z}_M$, respectively. Thus, in the undersampled case $N < M$ it is preferable to use the SDP formulation in \eqref{eq:sdp1}, while in the oversampled case $N \geq M$ it is preferable to apply the SDP formulation in \eqref{eq:sdp1b}. We remark that the subspace matching approach discussed in \cite{
Malioutov:LassoDoa, steffens2016noncircular} can be applied to formulation \eqref{eq:sdp1} as well. A further investigation of the subspace matching approach is, however, beyond the scope this paper. 
The various equivalent SDP implementations and the corresponding number of variables and constraints are listed in Table \ref{tab:sdp}.


\section{Gridless \mname{} Implementation} \label{sec:Gridless}
\noindent 
While the SDP formulations in Section \ref{sec:Implementation} are applicable to arbitrary array topologies, we consider in the following the special case of linear subarrays with a common baseline, where the sensors within each subarray are located at integer multiples of a baseline $\delta$, i.e., $\rho_m^{(p)} = \delta d_m^{(p)} $ with $d_m^{(p)} \in \mathbb{Z}$ for $m=1,\ldots,M_P$ and $p=1,\ldots,P$. This type of array topologies admits the extension of the \mname{} formulation to gridless frequency estimation. 

We start noting that strong duality holds for problem \eqref{eq:sdp1b} and consider the Lagrange dual problem, which is given as
\begin{subequations}
\label{eq:smr_dual}
\begin{align}
  \max_{\mb{\varUpsilon}_{1}, \mb{\varUpsilon}_{0}} & \; 
  - 2 \,\real \{ \tr(\mb{\varUpsilon}_{1}) \}  - \lambda \tr(\mb{\varUpsilon}_{0})  \label{eq:sdp_dual} \\
  \text{s.t.} \; & \mtx{ \hat{\mb{R}} & \mb{\varUpsilon}_{1} \\ \mb{\varUpsilon}_1^\tH & \mb{\varUpsilon}_0 } \succeq \mb{0} \\
			   & \mb{I}_P - \mb{B}^\tH(\nu_k) \, \mb{\varUpsilon}_0 \mb{B}(\nu_k) \succeq \mb{0}, \; k=1,\ldots, K , \label{eq:smr_dual_c}
\end{align}
\end{subequations}
where $\mb{\varUpsilon}_0$ is an $M \times M$ positive semidefinite matrix and $\mb{\varUpsilon}_1$ is of size $M \times M$ and does not exhibit specific structure. Complementary slackness requires that
\begin{align}
  \tr \big( \mb{S}_k (\mb{I}_P - \mb{B}^\tH(\nu_k) \, \mb{\varUpsilon}_0 \mb{B}(\nu_k) ) \big) = 0 ,
  \label{eq:compSlack1}
\end{align}
for $k=1,\ldots,K$, i.e., if $\mb{S}_k \neq \mb{0}$ then $\mb{I}_P - \mb{B}^\tH(\nu_k) \, \mb{\varUpsilon}_0 \mb{B}(\nu_k) $ must be singular, such that
\begin{align}
  \det \big(\mb{I}_P - \mb{B}^\tH(\nu_k) \, \mb{\varUpsilon}_0 \mb{B}(\nu_k) \big)
  \begin{cases}
	= 0 \quad \text{if } \mb{S}_k \neq \mb{0} \\
	\geq 0 \quad \text{if } \mb{S}_k = \mb{0} .
  \end{cases}
  \label{eq:compSlack2}
\end{align}
Condition \eqref{eq:compSlack2} indicates that instead of solving the primal problem \eqref{eq:smr1} and identifying the block-support from $\mb{S}$, we can equivalently solve the dual problem \eqref{eq:smr_dual} and identify the block-support from the roots of \eqref{eq:compSlack2}.

Let us consider the limiting case of an infinitesimal frequency grid spacing, i.e., $\lim_{K \rightarrow \infty} \nu_k-\nu_{k-1} = 0$ for $k=2,\ldots,K$, such that the frequency becomes a continuous parameter $\nu$. By introducing the variable
\begin{align}
  z = \te^{\tj \pi \nu \delta} ,
\end{align}
the subarray steering matrices in \eqref{eq:subarraySteeringMatrix} can be equivalently described as
\begin{align}
  \mb{B}(z)  = 
  \blkdiag \big( \mb{b}^{(1)}(z), \ldots, \mb{b}^{(P)}(z) \big) ,
  \label{eq:subarraySteeringMatrixZ}
\end{align}
where the subarray steering vectors are given as
\begin{equation}
  \mb{b}^{(p)} \left( z \right) =  \big[ 1, z^{d_2^{(p)} }, \ldots, z^{d_{M_p}^{(p)} } \big]^\tT ,
  \label{eq:bkZ}
\end{equation}
with $d_{m}^{(p)} \in \mathbb{Z}$ for $m=1,\ldots,M_p$ and $p=1,\ldots,P$. By the definition in \eqref{eq:subarraySteeringMatrixZ}, the matrix product $\mb{B}^\tH(z) \mb{\varUpsilon}_0 \mb{B}(z)$ in constraint \eqref{eq:smr_dual_c} constitutes a trigonometric matrix polynomial of degree $D = \max_{p,m_p} d^{(p)}_{m_p} $, according to
\begin{align}
  \mb{M}(z) 
  = \mb{B}^\tH(z) \mb{\varUpsilon}_0 \mb{B}(z) \label{eq:matPoly1} 
  = \sum_{i=-D}^{D} \mb{K}_i \, z^{i} ,
\end{align}
with matrix coefficients $\mb{K}_i$ of size $P \times P$ \cite{Dumitrescu:1086500}. In the continuous case the constraint \eqref{eq:smr_dual_c} is replaced by the constraint
\begin{align}
  \mb{I}_q - \mb{B}^\tH(z) \mb{\varUpsilon}_0 \mb{B}(z) \succeq \mb{0} \label{eq:matIneq}
\end{align}
which provides an upper bound on the matrix polynomial \eqref{eq:matPoly1} and can be implemented by semidefinite programming, e.g., by the problem formulation \eqref{eq:smr_dual_gridless} derived in Appendix \ref{sec:MatPoly} or by other techniques discussed in \cite{Dumitrescu:1086500}. Once the continuous implementation of problem \eqref{eq:smr_dual} is solved, the spatial frequencies can be recovered by finding the roots for which the left-hand side of \eqref{eq:matIneq} becomes singular, e.g., by rooting the continuous counterpart of \eqref{eq:compSlack2} using the techniques discussed in \cite{pesavento2005fast}. 


\subsection{Related Work} \label{sec:RelatedWork}
\noindent 
In a recent work \cite{steffens2017shiftinvariance} it has been shown that PCAs composed of identical subarrays admit gridless compressed sensing by means of the SPARROW formulation \cite{steffens2016compact}. Interestingly, for the special case considered in this section of PCAs composed of linear subarrays with a common baseline (and possibly missing sensors in particular subarrays) the dual problem formulation \eqref{eq:smr_dual} can equivalently be derived by means of the SPARROW formulation.

Let us consider the $\ell_{2,1}$ mixed-norm minimization problem for FCAs given in \eqref{eq:mixedVectorNorm} which can equivalently be formulated as the SPARROW problem \cite{steffens2016compact}
\begin{align}
  \min_{\mb{S} \in \pdMat{1}{K}} &\;
  \tr \big( (\mb{A}(\mb{\nu},\mb{\eta})\mb{S} \mb{A}^\tH(\mb{\nu},\mb{\eta} ) + \tilde{\lambda} \mb{I})^{-1} \hat{\mb{R}} \big) + 
  \tr( \mb{S} ) \label{eq:smr} 
\end{align}
where $\tilde{\lambda} > 0$ is a regularization parameter and $\mb{S}=\diag(s_1, \ldots, s_K) \succeq \mb{0}$ is of size $K \times K$. Problem \eqref{eq:smr} can be formulated as the semidefinite program
\begin{subequations}
\label{eq:sdp0}
\begin{align}
  \min_{\mb{S}, \mb{Z}_M} & \; \tr( \mb{Z}_M \hat{\mb{R}} ) + \tr( \mb{S} ) \\
  \text{s.t.} \; & \mtx{ \mb{Z}_M & \mb{I} \\ \mb{I} & 
  \mb{A}(\mb{\nu},\mb{\eta})\mb{S} \mb{A}^\tH(\mb{\nu},\mb{\eta}) + \tilde{\lambda} \mb{I} } \succeq \mb{0} \\
  & \mb{S} \in \pdMat{1}{K} \nonumber .
\end{align}
\end{subequations}
The Lagrange dual problem of \eqref{eq:sdp0} is given as 
\begin{subequations}
\label{eq:sdp_dual2}
\begin{align}
  \max_{\tilde{\mb{\varUpsilon}}_{1}, \tilde{\mb{\varUpsilon}}_{0}} & \; 
  - 2 \,\real \{ \tr(\tilde{\mb{\varUpsilon}}_{1}) \}  - \tilde{\lambda} \tr(\tilde{\mb{\varUpsilon}}_{0})   \\
  \text{s.t.} \; & \mtx{ \hat{\mb{R}} & \tilde{\mb{\varUpsilon}}_{1} \\ \tilde{\mb{\varUpsilon}}_1^\tH & \tilde{\mb{\varUpsilon}}_0 } \succeq \mb{0} \\
			   & \; \mb{a}^\tH(\nu_k,\mb{\eta}) \tilde{\mb{\varUpsilon}}_0 \, \mb{a}(\nu_k,\mb{\eta}) \leq 1, \; k=1,\ldots, K 
			   \label{eq:sdp_dual2c},
\end{align}
\end{subequations}
and with strong duality it follows from complementary slackness that
\begin{align}
  1 - \mb{a}^\tH(\nu_k,\mb{\eta}) \tilde{\mb{\varUpsilon}}_0 \, \mb{a}(\nu_k,\mb{\eta}) 
  \begin{cases}
	= 0 \quad \text{if } s_n \geq 0 \\
	\geq 0 \quad \text{if } s_n = 0 .
  \end{cases}
  \label{eq:compSlack_v4}
\end{align}
As previously discussed in the context of condition \eqref{eq:compSlack2}, the support of the vector $\mb{s} = [s_1, \ldots, s_K]^\tT$, i.e., the spatial frequency estimates, can equivalently be identified by rooting the function in \eqref{eq:compSlack_v4}.

Making use of the notation $\mb{a}(\nu_k,\mb{\eta}) = \mb{B}(\nu_k) \mb{\varphi}(\nu_k,\mb{\eta})$, as introduced in \eqref{eq:a_mu}, condition \eqref{eq:compSlack_v4} can be rewritten as 
\begin{align}
  & 1 - \mb{a}^\tH(\nu_k,\mb{\eta}) \, \tilde{\mb{\varUpsilon}}_0 \, \mb{a}(\nu_k,\mb{\eta}) \notag \\ = \, &  
  1 - \mb{\varphi}^\tH (\nu_k,\mb{\eta}) \, \mb{B}^\tH(\nu_k)\,  \tilde{\mb{\varUpsilon}}_0 \, \mb{B}(\nu_k) \, \mb{\varphi}(\nu_k,\mb{\eta}) 
  \notag \\ =\,&
  1 - \tilde{\mb{\varphi}}^\tH (\nu_k,\mb{\eta}) \, \mb{B}^\tH(\nu_k) \, \mb{\varUpsilon}_0 \, \mb{B}(\nu_k) \, \tilde{\mb{\varphi}}(\nu_k,\mb{\eta}) 
  \notag  \\ =\,&
  \tilde{\mb{\varphi}}^\tH (\nu_k,\mb{\eta}) \, \big( \mb{I}_p - \mb{B}^\tH(\nu_k) \, \mb{\varUpsilon}_0 \, \mb{B}(\nu_k) \big) \, \tilde{\mb{\varphi}}(\nu_k,\mb{\eta}) \geq 0 ,
  \label{eq:compSlack_v6}
\end{align}
where 
\begin{align}
  \tilde{\mb{\varphi}}(\nu_k,\mb{\eta}) &= \mb{\varphi}(\nu_k,\mb{\eta}) / \|\mb{\varphi}(\nu_k,\mb{\eta})\|_2 \\
  \mb{\varUpsilon}_0 &= \|\mb{\varphi}(\nu_k,\mb{\eta})\|_2^2 \, \tilde{\mb{\varUpsilon}}_0 \label{eq:upsion_tilde}.
\end{align}
Condition \eqref{eq:compSlack_v6} is fulfilled if
\begin{align}
  \mb{I}_q - \mb{B}^\tH(\nu_k) \mb{\varUpsilon}_0 \mb{B}(\nu_k) \succeq \mb{0}, \label{eq:matIneq2}
\end{align}
which is identical to the constraint \eqref{eq:smr_dual_c} in problem \eqref{eq:smr_dual}. Replacing the constraint \eqref{eq:sdp_dual2c} in problem \eqref{eq:sdp_dual2} by the condition \eqref{eq:matIneq2} and further using \eqref{eq:upsion_tilde} and the substitutions $\lambda = \tilde{\lambda} / \|\mb{\varphi}(\nu_k,\mb{\eta})\|_2^2$ and $\mb{\varUpsilon}_1 = \tilde{\mb{\varUpsilon}}_1$ shows that for the PCA case the dual problem \eqref{eq:sdp_dual2} can be reformulated as the dual problem  \eqref{eq:smr_dual}. As demonstrated in the previous section, condition \eqref{eq:matIneq2} can be extended to an infinitesimal grid spacing, resulting in a matrix polynomial constraint, such that the resulting gridless estimation problem can be implemented by semidefinite programming (see Appendix \ref{sec:MatPoly}).

\section{Numerical Results}\label{sec:sims}
\noindent
For experimental performance evaluation of our proposed \mname{} method we compare its estimation performance to the state-of-the-art methods spectral RARE \cite{Gershman:Rare} and root-RARE \cite{Pesavento:Rare} as well as the Cram\'er-Rao bound (CRB) \cite{Gershman:Rare}. 

For all simulations we use circular complex Gaussian source signals $\mb{\varPsi}$ with covariance matrix $\tE(\mb{\varPsi} \mb{\varPsi}^\tH) = N \mb{I}$, if not specified otherwise. We further consider spatio-temporal white circular complex Gaussian sensor noise $\mb{N}$ with covariance matrix $\tE(\mb{N} \mb{N}^\tH) = \sigma^2 N \mb{I}$ and define the signal-to-noise ratio (SNR) as $\tsnr~=~1/\sigma^2$. The vector $\mb{r}^{(p)} = [r_1^{(p)}, \ldots, r_{M_p}^{(p)}]^\tT$ contains the global sensor positions $\smash{r_{m}^{(p)}}$ of subarray $p$, for $m=1,\ldots,M_p$, $p=1,\ldots,P$, expressed in half-wavelength, as defined in \eqref{eq:SenPos}. If not stated otherwise, we perform $T=1000$ Monte Carlo trials for each experimental setup and compute the statistical error.

The estimation performance of the \mname{} method strongly depends on proper selection of a regularization parameter $\lambda$. While regularization parameter selection is a research field of its own, in this paper we follow a heuristic approach and select the regularization parameter as
\begin{align}
  \lambda = \max_p \sigma \sqrt{M_p \log (M)} ,
  \label{eq:regParSel}
\end{align}
which has shown good estimation performance in all investigated scenarios. 

We remark that the RARE and \mname{} method make different assumptions on the availability of a-priori knowledge. While the RARE method requires knowledge of the number of source signals, the regularization parameter selection for the \mname{} method according to \eqref{eq:regParSel} requires knowledge of the noise power. However, since estimation of these parameters itself might affect the frequency estimation performance of the RARE and \mname{} methods, we apply the standard assumption of perfectly known number of source signals and noise power and investigate the achievable performance under these idealized assumptions. 

\def\pw{10cm}
\def\ph{6cm}

\def\cGBC{blue}
\def\mGBC{o}
\def\cGLC{red}
\def\mGLC{x}
\def\cSR{brown}
\def\mSR{triangle}
\def\cRR{green!80!black}
\def\mRR{square}

\subsection{Arbitrary Array Topologies and Grid-Based Estimation}
\noindent
In the first scenario we consider a PCA with a large aperture, composed of $M=11$ sensors which are partitioned in $P=4$ linear subarrays with 3,2,3, and 3 sensors, respectively. The sensor positions for each subarray are $\mb{r}^{(1)}=[0.0, 0.6, 2.3]^\tT$, $\mb{r}^{(2)}=[12.2, 13.0]^\tT$, $\mb{r}^{(3)}=[21.5, 22.8, 23.6]^\tT$, and $\mb{r}^{(4)}=[37.6, 38.5, 41.1]^\tT$, and we assume no additional gain/phase offsets among the subarrays, i.e., $\mb{\alpha}=[1, 1, 1, 1]^\tT$ in \eqref{eq:Phik}. We further consider $L=3$ uncorrelated Gaussian source signals with spatial frequencies $\mb{\mu} = [0.5011, 0.4672, -0.2007]^\tT$. 

The array topology does not admit a direct implementation of the gridless \mname{} and the root-RARE methods such that we limit the experiments in this subsection to the investigation of the grid-based \mname{} method and the spectral RARE method. For both grid-based methods we use a gird of $K=400$ grid points according to $\mb{\nu}=[-1.000, -0.995, -0.999, \ldots, 0.995]^\tT$.

To investigate the frequency estimation performance, we compute the root-mean-square error of the frequency estimates in $\hat{\mb{\mu}}$ as
\begin{align}
  \trmse (\hat{\mb{\mu}})= \sqrt{\frac{1}{L T} \sum_{t=1}^T \sum_{l=1}^L \big| \mu_l - \hat{\mu}_l(t) \big|_{\rm wa}^2}, 
  \label{eq:rmse}
\end{align}
where $\hat{\mu}_l(t)$ denotes the frequency estimate of signal $l$ in trial $t$ and $|\mu_1-\mu_2|_{\rm wa}=\min_{i \in \mathbb{Z}}|\mu_1-\mu_2+2i|$ denotes the wrap-around distance for two frequencies $\mu_1,\mu_2 \in [-1, 1)$. Since the $\trmse$ computation \eqref{eq:rmse} requires the number of estimated source signals $\hat{L}$ to be equal to the true number of source signals $L$, we have to consider two special cases: in the case of overestimation of the model order, $\hat{L} > L$, we select the $L$ frequency estimates with the largest corresponding magnitudes, whereas we select $L-\hat{L}$ additional random spatial frequencies in the case of underestimation $\hat{L} < L$.

\begin{figure}[t]
\begin{center}
  \small
  \input{grid_rmse_fs_NSnp}
  \caption{Frequency estimation performance for a PCA of $M=11$ sensors in $P=4$ subarrays, with $\tsnr=6\tdB$ and varying number of snapshots $N$}
  \label{fig:grid_rmse_fs_NSnp}
  \vspace{.4cm}

  \input{grid_rmse_fs_SNR}
  \caption{Frequency estimation performance for a PCA of $M=11$ sensors in $P=4$ subarrays, with $N=20$ snapshots and varying $\tsnr$}
  \label{fig:grid_rmse_fs_SNR}
  \vspace{-.3cm}
\end{center}  
\end{figure}

In the first experiment the signal-to-noise ratio ($\tsnr$) is fixed to $\tsnr=6\tdB$, while the number of snapshots $N$ is varied. Figure \ref{fig:grid_rmse_fs_NSnp} clearly demonstrates that our proposed grid-based \mname{} technique outperforms the spectral RARE for low number of signal snapshots $N$. While the spectral RARE method is not able to always resolve the two closely spaced signals with spatial frequencies $\mu_1 = 0.5011$ and $\mu_2 = 0.4672$ for $N \leq 500$ signal snapshots, our proposed \mname{} method resolves the signals for any $N \geq 30$ snapshots. 

In a second experiment we fix the number of snapshots as $N=20$ and vary the SNR. As can be observed from Figure~\ref{fig:grid_rmse_fs_SNR} the grid-based \mname{} method shows superior threshold performance as compared to the spectral RARE. While the spectral RARE can reliably resolve the two closely spaced sources only for $\tsnr \geq 22 \tdB$, our proposed \mname{} can do so for $\tsnr \geq 8 \tdB$. For high $\tsnr$, spectral RARE reaches a bias in the RMSE which is caused mainly by the finite grid. A similar bias effect can be observed for the grid-based \mname{} method. However, for the grid-based \mname{} method the bias is larger than for the spectral RARE method and it is not only caused by the finite grid, as discussed in the following subsection.

\subsection{Resolution Performance and Estimation Bias}

\def\pIdx{1,2,3,5,8,12,15,20,25,30,37,46,54}

\begin{figure}
\begin{center}
  \small
  \input{sep1_rmse_fs}
  \caption{Frequency estimation performance for uniform linear PCA of $M=9$ sensors in $P=3$ linear subarrays, for $N=20$ snapshots and $\tsnr=0\tdB$}
  \label{fig:sep1_rmse_fs} 
  \vspace{.4cm}
  
  \input{sep2_rmse_fs}
  \caption{Frequency estimation performance for uniform linear PCA of $M=9$ sensors in $P=3$ linear subarrays, for $N=50$ snapshots and $\tsnr=20\tdB$}
  \label{fig:sep2_rmse_fs}  
  \vspace{.4cm}
  
  \input{sep2_bias_fs}
  \caption{Frequency estimation bias for uniform linear PCA of $M=9$ sensors in $P=3$ linear subarrays, for $N=50$ snapshots and $\tsnr=20\tdB$}
  \label{fig:sep2_bias_fs}
  \vspace{-.4cm}  
\end{center}  
\end{figure}

\noindent
For further investigation of the spatial frequency estimation bias, we consider a uniform linear array of $M=9$ sensors, partitioned into $P=3$ identical, uniform linear subarrays of 3 sensors each, without additional gain/phase offsets, i.e., $\mb{\alpha} = [1, 1, 1]^\tT$ in \eqref{eq:Phik}. For the experiment we consider $L=2$ uncorrelated signals and fix the spatial frequency of the first signal as $\mu_1=0.505$ while the spatial frequency of the second signal is varied according to $\mu_2 = \mu_1-\Delta \mu$ with $10^{-2} \leq \Delta \mu \leq 1$. For all grid-based estimation methods we make use of a uniform grid of $K=200$ points according to $\mb{\nu}=[-1, -0.99, -0.98, \ldots, 0.99]^\tT$. The SNR and number of snapshots are fixed as $\tsnr=0\tdB$ and $N=20$.

First, we observe from Figure \ref{fig:sep1_rmse_fs} that the spectral RARE performs significantly worse than root-RARE in terms of threshold performance, i.e., the spectral RARE cannot always resolve the two signals for a frequency separation of $\Delta \mu \lessapprox 0.4$ while the root-RARE can resolve the signals for $\Delta \mu \gtrapprox 0.12$. The reason for this difference in resolution performance is that the root-RARE method locates the roots of the corresponding matrix polynomial in the entire complex plane, while the spectral RARE only searches minima on the unit circle (see also \cite{barabell1983rootMusic, pesavento2000unitaryRootMusic}). In contrast to that, the grid-based and the gridless \mname{} methods both show rather similar estimation performance, comparable to that of the root-RARE method, and reach the CRB for sufficiently large frequency separation. This observation can be explained by the fact that both dual \mname{} optimization problems provide matrix polynomials with the roots 
of interest constrained on the unit circle, as discussed in Section \ref{sec:Gridless}. This explains the similar performance results of the grid-based and gridless \mname{} methods. The only difference between the grid-based and gridless \mname{} methods is that in the first case the roots are generated on a grid of candidate frequencies on the unit circle, while in the latter case the roots are continuously located on the unit circle. 

In a slightly modified experiment we fix the SNR and number of snapshots to $\tsnr = 20 \tdB$ and $N=50$, respectively. While the root-RARE method performs close to the CRB for the region of interest, the spectral RARE can not always resolve the signals for $\Delta \mu \lessapprox 0.06$ and reaches an estimation bias for large source separation, which is caused by the finite frequency grid. Furthermore, it can be observed that the estimation performance of the \mname{} methods deviates from that of the root-RARE method. For large frequency separation $\Delta \mu \gtrapprox 0.2$ the grid-based \mname{} method reaches the grid bias, similar to the spectral RARE. However, also for low frequency separation $\Delta \mu \lessapprox 0.2 $ both methods do not reach the CRB. This can be explained by an inherent frequency estimation bias for SR methods (see also \cite{Malioutov:LassoDoa, steffens2016compact}). For further investigation we compute the spatial frequency estimation bias as
\begin{align}
  \text{Bias} (\hat{\mb{\mu}}) = \sqrt{\frac{1}{L} \sum_{l=1}^L \left( \mu_l - {\rm Mean} (\hat{\mu}_l) \right)^2 },
  \label{eq:bias}
\end{align}
where the mean estimate for spatial frequency $\mu_l$ is computed as ${\rm Mean} (\hat{\mu}_l) = 1/T \sum_{t=1}^T \hat{\mu}_l(t)$.

For the given scenario, the estimation bias is displayed in Figure \ref{fig:sep2_bias_fs}. In the case of low frequency separation $\Delta \mu \lessapprox 0.2$ both \mname{} methods show a relatively large estimation bias of $\text{Bias} (\hat{\mb{\mu}}) \approx 0.01$. For larger frequency separation $\Delta \mu \gtrapprox 0.2$, the bias of the grid-based \mname{} method is mainly determined by the finite grid, while the bias of the gridless \mname{} method shows to be periodic in $\Delta \mu$. In difficult scenarios, with low SNR and low number of snapshots as for the previous setup, the estimation bias is below the CRB, such that it is negligible in the RMSE performance. The frequency estimation bias is a well known phenomenon in SR research \cite{Malioutov:LassoDoa, steffens2016compact} and bias mitigation techniques have been discussed, e.g., in \cite{Abramovich:EL}.

\subsection{Correlated Signals}
\noindent
As discussed in the previous subsection, gridless \mname{} and root-RARE show approximately equal resolution performance for uncorrelated signals in difficult scenarios with low SNR, low number of snapshots and uncorrelated signals. This situation changes in the case of correlated signals, where preprocessing in form of subspace separation, as required for the RARE method, becomes difficult. For further investigation of this aspect we consider a PCA of $M=9$ sensors partitioned into $P=3$ subarrays of 3,4 and 2 sensors with positions $\mb{r}^{(1)}=[0, 1, 3]^\tT$, $\mb{r}^{(2)}=[17.4,18.4, 19.4, 21.4]^\tT$ and $\mb{r}^{(3)}=[24.8, 25.8]^\tT$. Furthermore we consider gain/phase offsets among the subarrays according to $\mb{\alpha} = [1, \, 0.7 \cdot \te^{\tj \frac{2}{3}\pi}, \, 1.2 \cdot \te^{\tj \frac{1}{4}\pi} ]^\tT$ in \eqref{eq:Phik}. The SNR and number of snapshots are selected as $\tsnr = 0 \tdB$ and $N=30$. We consider $L=2$ source signals with spatial frequencies $\mb{\mu} = [0.505,\, 0.105]^\tT$ and 
a source covariance matrix given as
\begin{align}
  \tE = N \mtx{1 & \rho \\ \rho^{*} & 1} ,
\end{align}
where the correlation coefficient $\rho$ is assumed to be real-valued and varied in the experiment. For the grid-based estimation methods we consider a grid of $K=200$ candidate frequencies, defined as in the previous subsection. As seen from Figure \ref{fig:corr_rmse_fs}, the spectral and root-RARE methods fail to properly estimate the spatial frequencies for high correlation ($\rho > 0.6$) while the grid-based and gridless \mname{} methods still show estimation performance close to the CRB, since these methods do not require subspace separation. 

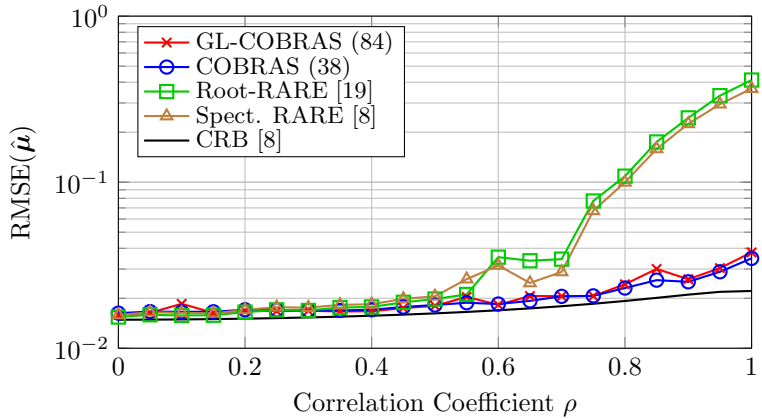
\begin{figure}[t]
\begin{center}
  \small
  \input{corr_rmse_fs}
  \caption{Frequency estimation performance for PCA of $M=9$ sensors in $P=3$ subarrays for $\tsnr=0 \tdB$, $N=30$ snapshot and $L=2$ source signals with varying real-valued correlation coefficient $\rho$ }
  \label{fig:corr_rmse_fs}
  \vspace{-.3cm}
\end{center}  
\end{figure}

\subsection{Array Calibration Performance}

\noindent
Besides estimation of the spatial frequencies, the \mname{} method also admits estimation of the subarray shifts in $\mb{\varphi}$ as defined in \eqref{eq:Phik}. Since the RARE methods do not provide direct estimation of the subarray shifts, we use the method presented in \cite{Parvazi:PCA} in equation (11) \footnote{Without the restriction that the complex phase terms must be of unit magnitude.} on the basis of the spatial frequency estimates obtained by the RARE methods. 

\begin{figure}[t]
\begin{center}  
  \small
  \input{cal_rmse_fs} 
  \caption{Frequency estimation performance for PCA of $M=10$ sensors in $P=4$ subarrays, for $N=20$ snapshot and $L=2$ uncorrelated source signals}
  \label{fig:cal_rmse_fs}
  \vspace{.4cm}
  
  \input{cal_rmse_phi}
  \caption{Displacement phase estimation performance for PCA of $M=10$ sensors in $P=4$ subarrays, for $N=20$ snapshot and $L=2$ uncorrelated source signals}
  \label{fig:cal_rmse_phi}
  \vspace{-.3cm}
\end{center}  
\end{figure}

The setup under investigation consists of a PCA of $M=10$ sensors partitioned into $P=4$ subarrays of 3,2,3 and 2 sensors at positions $\mb{r}^{(1)}=[0,2,3]^\tT$, $\mb{r}^{(2)}=[10.1,\,11.1]^\tT$, $\mb{r}^{(3)}=[27.4,\,28.4, \,30.4]^\tT$ and $\mb{r}^{(4)}=[54.8,\,56.8]^\tT$. The subarray gain/phase offsets are set as $\mb{\alpha} = [1,\, 1.3 \cdot \te^{\tj \frac{2}{3}\pi},\, 0.7 \cdot \te^{-\tj \frac{1}{4}\pi},\, 0.9 \cdot \te^{-\tj \frac{3}{5}\pi} ]^\tT$ in \eqref{eq:Phik}. We consider $L=3$ uncorrelated source signals with spatial frequencies $\mb{\mu} = [0.605,\, 0.255,\, -0.305]^\tT$ and the number of snapshot is set to $N=20$. 

Figure \ref{fig:cal_rmse_fs} displays the frequency estimation error of the different methods for varying SNR, where both \mname{} methods show the best thresholding performance but reach an estimation bias for $\tsnr \geq 15 \tdB$. Similarly, the spectral RARE algorithm reaches an estimation bias which is caused by the finite grid. On the other hand, the root-RARE performs asymptotically optimal and reaches the CRB for high SNR. The corresponding subarray shift estimation performance is displayed in Figure \ref{fig:cal_rmse_phi}, where the root-mean-square error is computed according to
\begin{align}
  \trmse (\hat{\mb{\varphi}})= \sqrt{\frac{1}{L T (P-1)} \sum_{t=1}^T \sum_{l=1}^L \big\| \mb{\varphi}_l - \hat{\mb{\varphi}}_l(t) \big\|_2^2 },
  \label{eq:rmse_phi}
\end{align}
with $\hat{\mb{\varphi}}_l(t)$ being the displacement phase vector estimate for signal $l$ in Monte Carlo trial $t$. As can be observed from Figure~\ref{fig:cal_rmse_phi}, the subarray shift estimation method in \cite{Parvazi:PCA}, based on the frequency estimates obtained from the RARE methods, achieves a relatively large estimation bias for high SNR. In contrast to that, the grid-based and gridless \mname{} methods show a significantly reduced estimation error, which demonstrates the advantage of joint frequency and displacement phase estimation.

\subsection{Computational Complexity}

\begin{figure}
\begin{center}
  \small
  \input{comp_time_NSnp} 
  \caption{Average computation time of different SDP implementations for uniform linear PCA with $M=9$ sensors in $P=3$ subarrays, with $K=100$ grid points and varying number of snapshots $N$}
  \label{fig:comp_time_NSnp}  
  \vspace{.4cm}
  
  \input{comp_time_NGrd}
  \caption{Average computation time of different SDP implementations for uniform linear PCA with $M=9$ sensors in $P=3$ subarrays, with $N=9$ signal snapshots and varying grid size $K$}
  \label{fig:comp_time_NGrd}  
  \vspace{-.3cm}  
\end{center}  
\end{figure}

\noindent
To investigate the computation time of the \mname{} formulation, we perform simulations in Matlab using the SeDuMi solver \cite{S98guide} with the CVX interface \cite{grant2008,grant2014} on a machine with an Intel Core i5-760 CPU @ $2.80\;{\rm GHz} \times 4$ and $8\,{\rm GByte}$ RAM. We consider a scenario with two independent complex Gaussian sources with static spatial frequencies $\mu_1=0.505$ and $\mu_2=-0.205$ and a uniform linear PCA of $M=9$ sensors partitioned into $P=3$ identical and uniform linear subarrays of 3 sensors. We neglect subarray gain/phase offsets, i.e., $\mb{\alpha}=[1, 1, 1]^\tT$ in \eqref{eq:Phik}. 

For the first experiment the SNR is fixed at $\tsnr=0\tdB$ while the number of snapshots $N$ is varied. Figure \ref{fig:comp_time_NSnp} shows the average computation time for $T=100$ Monte Carlo runs of the SDP implementation of the $\ell_{*,1}$ mixed-norm minimization \eqref{eq:NucNormSdp2} and the grid-based \mname{} formulations \eqref{eq:sdp1} and \eqref{eq:sdp1b} with a grid size of $K=100$, as well as the gridless (GL-) \mname{} formulation in \eqref{eq:smr_dual_gridless}. The computation time is measured only for solving the corresponding optimization problem in CVX. Pre-processing steps, such as computation of the sample covariance matrix, or post-processing steps, such as peak-search or polynomial rooting, are not included into this consideration. As can be observed from Figure \ref{fig:comp_time_NSnp}, for a number of $N < 5$ snapshots all grid-based methods exhibit approximately equal computation time. For $5 \leq N < 40 $ the $\ell_{*,1}$ mixed-norm minimization problem has largest computation 
time while the \mname{} formulation \eqref{eq:sdp1} requires longest computation time for $N > 40$, due to the large dimension of the semidefinite constraint \eqref{eq:sdp1Con1}. Regarding the computation time of the grid-based \mname{} formulation using the sample covariance matrix \eqref{eq:sdp1b} we observe that it is relatively constant for any number of snapshots $N$ and lower than for the other implementations especially for large number of snapshots $N > 10$. The lowest computation time is required for solving the GL-\mname{} implementation \eqref{eq:smr_dual_gridless}. 

Figure \ref{fig:comp_time_NGrd} shows the average computation time for $T=100$ Monte Carlo runs for a varying number of grid points $K$ and a fixed number of $N=9$ signal snapshots, corresponding to the case where the signal subspace matching techniques are applied (compare \cite{Malioutov:LassoDoa, steffens2016noncircular}). For all grid-based methods the number of SDP constraints grows nearly linear with the number of grid points $K$. Clearly, the $\ell_{*,1}$ mixed-norm minimization approach has the largest computation time for any investigated number of grid points $K$. The grid-based \mname{} formulations \eqref{eq:sdp1} and \eqref{eq:sdp1b} show approximately equal computation time, since both formulations have SDP constraints of identical dimension. Since the GL-\mname{} formulation is independent of the grid size $K$, it is constant for all grid size numbers $K$ in Figure \ref{fig:comp_time_NGrd} and provides the fastest computation of all methods under investigation.

\section{Conclusion}\label{sec:conclusion}
\noindent
Partly calibrated arrays are attractive setups for direction of arrival estimation. Computationally efficient subspace-based methods such as RARE and ESPRIT show asymptotically optimal performance, but have problems in difficult scenarios such as low sample size, low signal-to-noise ratio or correlated source signals. Sparse recovery in the form of $\ell_{*,1}$ mixed-norm minimization has been shown to be an attractive alternative to subspace-based methods in these difficult scenarios. In this paper we have derived a compact, equivalent formulation of the $\ell_{*,1}$ mixed-norm minimization, referred to as COmpact Block- and RAnk-Sparse recovery (\mname{}). The \mname{} formulation is attractive especially in the case of a large number of signal snapshots. For the special case of subarrays with common baseline we have presented an extension to gridless estimation, referred to as gridless \mname{} (GL-\mname{}). 

As shown by numerical results, the grid-based \mname{} significantly outperforms the spectral RARE method in terms of thresholding performance for closely spaced source signals. Furthermore, the \mname{} method outperforms the RARE method in the case of strongly correlated source signals and in the calibration (subarray shift estimation) performance. A drawback of the $\ell_{*,1}$ mixed-norm minimization approach and the \mname{} formulation is the estimation bias which becomes significant in the asymptotic case of large number of snapshots or high signal-to-noise ratio. However, if higher estimation accuracy is required, the \mname{} estimates can be used to provide initial estimates, e.g., for a subsequent maximum likelihood estimator. 

\section*{Acknowledgment}
This work was supported by the German Research Foundation (DFG) within the DFG priority program on Compressed Sensing in Information Processing (CoSIP DFG-SPP 1798).

\refstepcounter{AppendixCount} \label{sec:proof}
\section*{Appendix \Alph{AppendixCount} - Proof of Problem Equivalence} 
\noindent
The proof of Theorem \ref{th:equivalence} relies on the following lemma, see \cite{Srebro2005, Recht2010}:
\begin{lemma} \label{NucNormMin}
The nuclear norm of the $P \times N$ matrix $\mb{Q}_k$ is equivalently computed by the minimization problem
\begin{align}
  \left\| \mb{Q}_k \right\|_{*} = 
  \min_{\mb{\varGamma}_k, \mb{G}_k} 
  \Big\{ \frac{1}{2} \big(\|\mb{\varGamma}_{k}\|_\tF^2 + \| \mb{G}_k\|_\tF^2 \big) : \, \mb{\varGamma}_{ k} \mb{G}_k = \mb{Q}_k \Big\} 
  \label{eq:NucNormMin}
\end{align}
where $\mb{\varGamma}_k$ and $\mb{G}_k$ are complex matrices of dimensions $P \times r$ and $r \times N$, respectively, with $r=\min(N,P)$. 
\end{lemma}

\begin{proof}[Proof of Lemma \ref{NucNormMin}]
Let us define the compact singular value decomposition of matrix $\mb{Q}_k$ as 
\begin{align}
  \mb{Q}_k = \mb{U}_k \mb{\varSigma}_k \mb{V}_k^\tH, 
  \label{eq:SvdQ}
\end{align}
such that the factorization terms of $\mb{Q}_k = \mb{\varGamma}_{k} \mb{G}_k$ can be expressed as
\begin{align}
  \mb{\varGamma}_k = \mb{U}_k \mb{\varPi}_{k,1} \mb{W}_k^\tH \quad \text{ and } \quad
  \mb{G}_k = \mb{W}_k \mb{\varPi}_{k,2} \mb{V}_k^\tH,
  \label{eq:FactorSvds}
\end{align}
for $\mb{\varSigma}_k = \mb{\varPi}_{k,1} \mb{\varPi}_{k,2}$ of size $r \times r$ and some $r \times r$ arbitrary unitary matrix $\mb{W}_k$, i.e., $\mb{W}_k^\tH \mb{W}_k = \mb{I}_r$. Based on \eqref{eq:FactorSvds} it holds that
\begin{align}
  \| \mb{Q}_k \|_{*} = \| \mb{\varSigma}_k \|_{*} &= \|  \mb{\varPi}_{k,1} \mb{\varPi}_{k,2} \|_{*} \notag \\
  & \leq \|  \mb{\varPi}_{k,1} \|_\tF \| \mb{\varPi}_{k,2} \|_\tF \notag \\
  & = \|  \mb{\varGamma}_{k} \|_\tF \| \mb{G}_{k,2} \|_\tF ,
\end{align}
where the inequality stems from the Cauchy-Schwartz inequality and is fulfilled with equality if and only if $\mb{\varPi}_{k,1} = \mb{\varPi}_{k,2} = \mb{\varSigma}^{1/2}_k$. In this case, the matrix factors in \eqref{eq:FactorSvds} are given as
\begin{align}
  \mb{\varGamma}_k = \mb{U}_k \mb{\Sigma}_k^{\frac{1}{2}} \mb{W}_k^\tH \quad \text{ and } \quad
  \mb{G}_k = \mb{W}_k \mb{\Sigma}_k^{\frac{1}{2}} \mb{V}_k^\tH .
  \label{eq:OptFactorSvds}
\end{align}
Furthermore, by the arithmetic-geometric-mean inequality it follows that
\begin{align}
  \|  \mb{\varGamma}_{k} \|_\tF \| \mb{G}_{k,2} \|_\tF \leq \frac{1}{2} \big(\|\mb{\varGamma}_k\|_\tF^2 + \| \mb{G}_k\|_\tF^2 \big) ,
\end{align}
where equality holds if $\|\mb{\varGamma}_k\|_\tF = \| \mb{G}_k\|_\tF$, such that the minimum of \eqref{eq:NucNormMin} is given by $\left\| \mb{Q}_k \right\|_{*} = \frac{1}{2} \big(\|\mb{\varGamma}_{ k}\|_\tF^2 + \| \mb{G}_k\|_\tF^2 \big)$ with $\mb{\varGamma}_{ k}$ and $\mb{G}_{ k}$ given by \eqref{eq:OptFactorSvds}. 


\end{proof}

\begin{proof}[Proof of Theorem \ref{th:equivalence}]
Based on Lemma \ref{NucNormMin}, the $\ell_{*,1}$ mixed-norm of the source signal matrix $\mb{Q} = [\mb{Q}_1^\tT, \ldots, \mb{Q}_K^\tT]^\tT$, as defined in \eqref{eq:Nuc1MixedNorm}, is equivalently computed by 
\vspace{.1cm}
\begin{align}  
  \| \mb{Q} \|_{*,1} &= \smash{ \sum_{k=1}^{K} } \left\| \mb{Q}_k \right\|_{*} \notag \\[.7em] 
  &= \min_{\{ \mb{\varGamma}_k, \mb{G}_k\} }  
  \Big\{ \frac{1}{2} \smash{ \sum_{k=1}^K } \big(\|\mb{\varGamma}_{ k}\|_\tF^2 + \| \mb{G}_k\|_\tF^2 \big) : \, \mb{\varGamma}_{ k} \mb{G}_k = \mb{Q}_k \Big\} \notag \\[.5em]
  &= \min_{\mb{\varGamma} \in \dMat{P}{r}{K}, \mb{G}} \; 
  \Big\{ \frac{1}{2} (\| \mb{\varGamma} \|_\tF^2 + \| \mb{G}\|_\tF^2) : \mb{Q} = \mb{\varGamma} \mb{G} \Big\}
  \label{eq:lNuc1FacConst}
\end{align}
where $r=\min(N,P)$, $\mb{\varGamma} = \blkdiag(\mb{\varGamma_1}, \ldots, \mb{\varGamma}_K)$ is taken from the set $\dMat{P}{r}{K}$ of block-diagonal matrices composed of $K$ blocks of size $P \times r$ on the main diagonal, and $\mb{G}=[\mb{G}_1^\tT, \ldots, \mb{G}_K^\tT]^\tT$ is a $(K r) \times N$ complex matrix composed of blocks $\mb{G}_k$, for $k=1,\ldots,K$. Inserting equation \eqref{eq:lNuc1FacConst} into the $\ell_{*,1}$ mixed-norm minimization problem in \eqref{eq:mixedVectorNorm_v2} we formulate the minimization problem
\begin{equation}  
  \min_{\substack{\mb{\varGamma} \in \dMat{P}{r}{K}, \mb{G}}}
  \frac{1}{2} \left\| \mb{B}\mb{\varGamma} \mb{G} - \mb{Y} \right\|_\tF^2 + \frac{\lambda \sqrt{N}}{2} (\| \mb{\varGamma} \|_\tF^2 + \| \mb{G}\|_\tF^2) 
  \label{eq:bilinOpt} .
\end{equation}
For a fixed matrix $\mb{\varGamma}$, the minimizer $\opt{\mb{G}}$ of problem \eqref{eq:lNuc1FacConst} has the closed form expression
\begin{align}
  \opt{\mb{G}} 
  &= ( \mb{\varGamma}^\tH \mb{B}^\tH \mb{B} \mb{\varGamma} + \lambda \sqrt{N} \mb{I} )^{-1} \mb{\varGamma}^\tH \mb{B}^\tH \mb{Y} 
  \nonumber \\ &=
  \mb{\varGamma}^\tH \mb{B}^\tH  ( \mb{B} \mb{\varGamma} \mb{\varGamma}^\tH \mb{B}^\tH + \lambda \sqrt{N} \mb{I} )^{-1} \mb{Y}
  \label{eq:optG}
\end{align}
where the last equation is derived from the Woodbury matrix identity \cite[p.151]{searle1982}. Reinserting the optimal matrix $\opt{\mb{G}}$ into equation \eqref{eq:bilinOpt} and using basic reformulations of the objective function results in the concentrated minimization problem
\begin{align}
  \smash{\min_{\mb{\varGamma} \in \dMat{P}{r}{K}}} \frac{\lambda \sqrt{N} }{2} \Big(
  \text{Tr}\big( ( \mb{B} \mb{\varGamma} \mb{\varGamma}^\tH \mb{B}^\tH + \lambda \sqrt{N} \mb{I} )^{\scalebox{0.75}[1.0]{-}1} \mb{Y} \mb{Y}^\tH \big)
  + \text{Tr} \big( \mb{\varGamma} \mb{\varGamma}^\tH \big) \Big) .
  \label{eq:equivalenceStep3}
\end{align}
Upon summarizing $\mb{Y} \mb{Y}^\tH / N = \hat{\mb{R}}$ and defining the positive semidefinite block-diagonal matrix 
\begin{align}
  \mb{S}=\mb{\varGamma} \mb{\varGamma}^\tH / \sqrt{N} \in \pdMat{P}{K}
  \label{eq:SMatDef}
\end{align}
we can rewrite \eqref{eq:equivalenceStep3} as 
\begin{align}
  \min_{\mb{S} \in \pdMat{P}{K}} \;
  \frac{\lambda N}{2} \Big( \text{Tr}\big( ( \mb{B} \mb{S} \mb{B}^\tH + \lambda \mb{I} )^{-1} \hat{\mb{R}} \big) 
  & + \text{Tr} \big( \mb{S} \big) \Big) .
  \label{eq:equivalenceStep4}
\end{align}
Neglecting the factor $\lambda N/2$ in \eqref{eq:equivalenceStep4}, we arrive at formulation \eqref{eq:smr1}. Using equations \eqref{eq:SvdQ}, \eqref{eq:OptFactorSvds} and the definition of $\opt{\mb{S}} = \blkdiag(\opt{\mb{S}}_1, \ldots, \opt{\mb{S}}_K)$ in \eqref{eq:SMatDef} we conclude that
\begin{align}
  \opt{\mb{S}}_k = \frac{1}{\sqrt{N}} \opt{\mb{\varGamma}}_k \opt{\mb{\varGamma}}_k^{\tH} 
		  = \frac{1}{\sqrt{N}} ( \opt{\mb{Q}}_k \opt{\mb{Q}}_k^\tH )^{1/2}
\end{align}
as given in \eqref{eq:magIdentity}. Making further use of \eqref{eq:optG} and the factorization in \eqref{eq:lNuc1FacConst} we obtain
\begin{align}
  \opt{\mb{Q}} =& \opt{\mb{\varGamma}} \opt{\mb{G}} \nonumber \\
			   =&  \opt{\mb{\varGamma}} \opt{\mb{\varGamma}}^\tH \mb{B}^\tH  ( \mb{B} \opt{\mb{\varGamma}} \opt{\mb{\varGamma}}^\tH \mb{B}^\tH + \lambda \sqrt{N} \mb{I} )^{-1} \mb{Y} \nonumber \\
			   =&  \opt{\mb{S}} \mb{B}^\tH  ( \mb{B} \opt{\mb{S}} \mb{B}^\tH + \lambda \mb{I} )^{-1} \mb{Y}
\end{align}
which corresponds to relation \eqref{eq:smr2}. 

\end{proof}

\refstepcounter{AppendixCount} \label{sec:MatPoly}
\section*{Appendix \Alph{AppendixCount} - SDP Form of the Matrix Polynomial Constraint} 
\noindent
As discussed in Section \ref{sec:Gridless}, the matrix $\mb{M}(z)$ represents a matrix polynomial of degree $D$ in the variable $z$ \cite{Dumitrescu:1086500}. To see the relation between matrix $\mb{\varUpsilon}_0$ and the matrix coefficients $\mb{K}_i$, let us define the $((D+1)P) \times P$ matrix
\begin{align}
  \mb{\varOmega}(z) = \mtx{\mb{I}_P & z \mb{I}_P & z^2 \mb{I}_P & \ldots & z^D \mb{I}_P}^\tT 
  \label{eq:matZ}
\end{align}
and introduce the $M \times ((D+1)P)$ permutation and selection matrix $\mb{J}$ such that the subarray steering matrix can be expressed as
\begin{align}
  \mb{B}(z) = \mb{J} \mb{\varOmega}(z) \label{eq:PermJ} .
\end{align}
Inserting \eqref{eq:PermJ} in \eqref{eq:matPoly1} yields
\begin{align}
  \mb{M}(z) 
  &= \mb{B}^\tH(z) \mb{\varUpsilon}_0 \mb{B}(z) \notag \\
  &= \mb{\varOmega}^\tH(z) \mb{J}^\tH \mb{\varUpsilon}_0 \mb{J} \mb{\varOmega}(z) \notag \\
  &= \mb{\varOmega}^\tH(z) \mb{F} \mb{\varOmega}(z) , \label{eq:matPoly2} 
\end{align}
where $\mb{F} = \mb{J}^\tH \mb{\varUpsilon}_0 \mb{J}$ is of size $((D+1)P) \times ((D+1)P)$ and is composed of the $P \times P$ blocks $\mb{F}_{i,j}$, for $i,j=1,\ldots D+1$, as
\begin{align}
  \mb{F} = 
  \mtx{
	\mb{F}_{1,1} & \cdots & \mb{F}_{1,D+1} \\
	\vdots 		 & \ddots & \vdots 		  \\
	\mb{F}_{D+1,1} & \cdots & \mb{F}_{D+1,D+1} \\
  } .
\end{align}
Equation \eqref{eq:matPoly2} is also referred to as the Gram matrix representation of the polynomial $\mb{M}(z)$, and $\mb{F}$ is referred to as the corresponding Gram matrix \cite{Dumitrescu:1086500}.

We define the block trace operator for matrix $\mb{F}$ as
\begin{align}
  \blktr{P} (\mb{F}) = \sum_{i=1}^{D+1} \mb{F}_{i,i}, \label{eq:BlkTr}
\end{align}
i.e., the summation of the $P \times P$ submatrices $\mb{F}_{i,i}$, for $i=1,\ldots,D+1$, on the main diagonal of matrix $\mb{F}$. Furthermore, let us define the $(D+1) \times (D+1)$ elementary Toeplitz matrix $\mb{\varTheta}_i$, with ones on the $i$th diagonal and zeros elsewhere, as well as the elementary block Toeplitz matrix $\mb{\varXi}_i = \mb{\varTheta}_i \otimes \mb{I}_P$. 

Using the block trace operator \eqref{eq:BlkTr} and the elementary block Toeplitz matrices $\mb{\varXi}_i$, the matrix coefficients $\mb{K}_i$ in \eqref{eq:matPoly1} can be computed from the Gram matrix $\mb{F}$ in \eqref{eq:matPoly2} as
\begin{align}
  \mb{K}_i = \blktr{P} ( \mb{\varXi}_i \mb{F} ) ,
  \label{eq:GramMapping}
\end{align}
i.e., the summation of the $P \times P$ submatrices on the $i$th block-diagonal of the Gram matrix $\mb{F}$. Note that the mapping \eqref{eq:matPoly1} is unique, i.e., for any PCA steering matrix block $\mb{B}(z)$ and matrix $\mb{\varUpsilon}_0$ the coefficients $\mb{K}_i$, $i=1,\ldots,D$, of the matrix polynomial $\mb{M}(z)$ are unique. However, the Gram matrix $\mb{F}$ in \eqref{eq:matPoly2} is not unique, i.e., a matrix polynomial $\mb{M}(z)$ generally admits different Gram matrix representations. 

Let us define a second matrix polynomial which has constant value $\mb{I}_P$ as
\begin{align}
  \mb{\varOmega}^\tH(z) \mb{H} \mb{\varOmega}(z) = \mb{I}_P
  \label{eq:MatPoly3}
\end{align}
such that the corresponding Gram matrix $\mb{H}$ of size $((D+1)P) \times ((D+1)P)$ fulfills 
\begin{subequations}
\label{eq:MatPoly3b}
\begin{align}
  \blktr{P} (\mb{H}) &= \mb{I}_P, \\
  \blktr{P} ( \mb{\varXi}_i \mb{H} ) &= \mb{0} \tfor i \neq 0  .
\end{align}
\end{subequations}
By using \eqref{eq:matPoly2}, \eqref{eq:MatPoly3} and \eqref{eq:MatPoly3b} we can express the constraint \eqref{eq:matIneq} as
\begin{align}
  \mb{I}_P - \mb{B}(z)^\tH \, \mb{\varUpsilon}_0 \mb{B}(z) = 
  \mb{\varOmega}^\tH ( \mb{H} - \mb{J}^\tH \mb{\varUpsilon}_0 \, \mb{J} ) \mb{\varOmega} \succeq \mb{0}
\end{align}
which is fulfilled for  
\begin{align}
  \mb{H} - \mb{J}^\tH \mb{\varUpsilon}_0  \, \mb{J} \succeq \mb{0} .
  \label{eq:GramMatConst}
\end{align}
Applying \eqref{eq:MatPoly3b} and \eqref{eq:GramMatConst} in problem \eqref{eq:smr_dual} we can define the gridless frequency estimation problem
\begin{subequations}
\label{eq:smr_dual_gridless}
\begin{align}
  \max_{\mb{\varUpsilon}_{1}, \mb{\varUpsilon}_{0}, \mb{H}} & \; 
  - 2 \,\real \{ \tr(\mb{\varUpsilon}_{1}) \}  - \lambda \tr(\mb{\varUpsilon}_{0})  \label{eq:sdp_gls_dual} \\
  \text{s.t.} \; & \mtx{ \hat{\mb{R}} & \mb{\varUpsilon}_{1} \\ \mb{\varUpsilon}_1^\tH & \mb{\varUpsilon}_0 } \succeq \mb{0} \\
			   & \mb{H} - \mb{J}^\tH \mb{\varUpsilon}_0 \, \mb{J} \succeq \mb{0} \\
			   & \blktr{P} (\mb{H}) = \mb{I}_{P} \\
			   & \blktr{P} ( \mb{\varXi}_i \mb{H} ) = \mb{0} \tfor i \neq 0 .
\end{align}
\end{subequations}
Given a minimizer $\opt{\mb{\varUpsilon}}_0$ to problem \eqref{eq:smr_dual_gridless} the frequency estimation problem reduces to finding roots for which the constraint \eqref{eq:matIneq} becomes singular, as discussed in Section \ref{sec:Gridless}. 

\bibliographystyle{IEEEtran}
\bibliography{refferences}

\end{document}

%% file: texMacros.tex
\newcommand{\tT}{\text{\sf T}}
\newcommand{\tH}{\text{\sf H}}
\newcommand{\tF}{\text{\sf F}}
\newcommand{\tj}{\text{\rm j}}
\newcommand{\tfor}{\text{ for }}

\newcommand{\tst}{{\rm s.t.}}
\newcommand{\tsnr}{{\rm SNR}}
\newcommand{\trmse}{{\rm RMSE}}
\newcommand{\tdB}{{\rm \,dB}}
\newcommand{\tE}{\ensuremath{\text{E}}}
\newcommand{\te}{{\rm e}}

\newcommand{\diag}{{\rm diag}}
\newcommand{\blkdiag}{{\rm blkdiag}}
\newcommand{\blktr}[1]{{\rm blkTr}_{#1}}

\newcommand{\tr}{{\rm Tr}}
\newcommand{\rank}{{\rm rank}}
\newcommand{\real}{\ensuremath{\text{Re}}}

\newcommand{\mb}[1]{\ensuremath{\boldsymbol{#1}}}
\newcommand{\bmtx}{\ensuremath{\begin{bmatrix}}}
\newcommand{\emtx}{\ensuremath{\end{bmatrix}}}
\newcommand{\mtx}[1]{\ensuremath{\begin{bmatrix} #1 \end{bmatrix}}}

\newcommand{\opt}[1]{\ensuremath{\smash{\hat{#1}}}}

\newcommand{\pdMat}[2]{\mathcal{B}^{#1\,+}_{#2}}
\newcommand{\dMat}[3]{\mathcal{B}^{#1 \times #2}_{#3}}

\newcommand{\mname}{{COBRAS}}

\newtheorem{theorem}{Theorem}
\newtheorem{lemma}{Lemma}

\newtheorem{corollary}{Corollary}


%% file: tikzMacros.tex
\pgfplotsset{every axis legend/.append style={row sep=-3pt, inner ysep=0pt, inner xsep=2pt, font=\footnotesize}}

\tikzset{ 
    table/.style={
        matrix of nodes,
        row sep=-\pgflinewidth,
        column sep=-\pgflinewidth,
        nodes={
			minimum height=3em,
	        minimum width=\twb,		
            rectangle,
            draw=black,
		   	text=black,
            inner xsep=5pt,
            inner ysep=3pt
        },
        nodes in empty cells,
        column 1/.style={
            nodes={minimum width=\twa,		
				   fill=gray!10,
				   text badly ragged}
        },
        row 1/.style={
            nodes={
				minimum height=1.7em,
                fill=gray!10,
                font=\bfseries}
        }
    }
}

\tikzset{every picture/.append style={baseline=(current bounding box.north west)}}

\colorlet{cs0}{black!30!brown}
\colorlet{cs1}{black!30!red}
\colorlet{cs2}{white!30!blue}
\colorlet{cs3}{black!25!green}
\colorlet{cs4}{black!25!orange}

\tikzstyle{sensor}=[circle, fill=cs0, inner sep=0pt, minimum height=1.5mm]
\tikzstyle{source}=[fill=red, inner sep=0pt, minimum height=1.5mm, minimum width=1.5mm]
\tikzstyle{sensorBg}=[line width=3mm, line join=round, cap=round, fill=black!10, draw=black!10]

\pgfdeclarelayer{background}
\pgfdeclarelayer{foreground}
\pgfsetlayers{background,main,foreground}

\def\NGrd{12}
\def\dA{3}

\def\dC{10}



%% file: modelImg1.tex
\begin{tikzpicture}
	
	\def\disp{(1.0,0)}
	\tikzset{>=stealth'}
	
	\node[sensor, fill=cs1] (s1) at (-4.0,0) {};
	\node[sensor, fill=cs1] (s2) at ($(s1)+(0.9,0)$) {};
    
	\node[sensor, fill=cs2] (s3) at (-1.6,0) {};
	\node[sensor, fill=cs2] (s4) at ($(s3)+(0.7,0)$) {};
	\node[sensor, fill=cs2] (s5) at ($(s3)+(1.2,0)$) {};
	\node[sensor, fill=cs2] (s6) at ($(s3)+(2.3,0)$) {};
	
	\node[sensor, fill=cs3] (s7) at (3.1,0) {};
	\node[sensor, fill=cs3] (s8) at ($(s7)+(0.4,0)$) {};
	\node[sensor, fill=cs3] (s9) at ($(s7)+(1.2,0)$) {};
	
	\begin{pgfonlayer}{background} 
	  \draw[sensorBg] (s1.west) -- (s2.east);	
	  \draw[sensorBg] (s3.west) -- (s6.east);		
	  \draw[sensorBg] (s7.west) -- (s9.east);		
	\end{pgfonlayer}           
    
	\coordinate (h1) at (0,0.3);
	\draw[<->,gray] ($(s1)+(h1)$) -- ($(s3)+(h1)$) node[pos=1,above,black] {$\eta^{(2)}$};
	\draw[ ->,gray] ($(s3)+(h1)$) -- ($(s7)+(h1)$) node[pos=1,above,black] {$\eta^{(3)}$};
	
    \draw[<->,gray] ($(s1)-(h1)$) -- ($(s2)-(h1)$) node[pos=1,below,black] {$\rho_2^{(1)}$};
	
	\draw[<->,gray] ($(s3)-(h1)$) -- ($(s4)-(h1)$) node[pos=1,below,black] {$\rho_2^{(2)}$};
	\draw[ ->,gray] ($(s4)-(h1)$) -- ($(s5)-(h1)$) node[pos=1,below,black] {$\rho_3^{(2)}$};
	\draw[ ->,gray] ($(s5)-(h1)$) -- ($(s6)-(h1)$) node[pos=1,below,black] {$\rho_4^{(2)}$};
    
    \draw[<->,gray] ($(s7)-(h1)$) -- ($(s8)-(h1)$) node[pos=1,below,black] {$\rho_2^{(3)}$};
	\draw[ ->,gray] ($(s8)-(h1)$) -- ($(s9)-(h1)$) node[pos=1,below,black] {$\rho_3^{(3)}$};
	
    \node[circle, inner sep=1pt] (cs) at (0.0,1) {};
    \draw[gray] ($(cs)-(1,0)$) -- ++(2,0);
	
    \def\angA{180-\dA/\NGrd*180}
    \def\angC{180-\dC/\NGrd*180}
    
    \coordinate (src1) at (intersection cs: first line={(cs)--($(cs)+(\angA:6)$)}, second line={(0,2.5)--(3,2.5)}) {};
    \coordinate (src2) at (intersection cs: first line={(cs)--($(cs)+(\angC:6)$)}, second line={(0,2.5)--(3,2.5)}) {};
	
    \draw[gray,<-] ($(cs)+(\angA:0.8)$) arc (\angA:180:0.8cm);
    \node[anchor=south east] at ($(cs)+(-0.3,-0.05)$) {\textcolor{black}{$\theta_1$}};
    \draw[gray,->] ($(cs)+(180:1)$) arc (180:\angC:1cm);
    \node[anchor=south east] at ($(cs)+(0.5,0.5)$){\textcolor{black}{$\theta_2$}};
    
    \draw[->] (src1) -- (cs); 
    \draw[->] (src2) -- (cs); 
	
    \def\wcut{1.2mm}
    \def\hcut{.4mm}
    \coordinate (cut1) at ($(src1)!0.6cm!(cs)$);
    \draw[double distance=\hcut] ($ (cut1)!\wcut!90:(cs) $) -- ($ (cut1)!\wcut!-90:(cs) $);
    \draw[white] ($ (cut1)!\wcut!90:(cs) $) -- ($ (cut1)!\wcut!-90:(cs) $);
    \coordinate (cut2) at ($(src2)!0.6cm!(cs)$);
    \draw[double distance=\hcut] ($ (cut2)!\wcut!90:(cs) $) -- ($ (cut2)!\wcut!-90:(cs) $);
    \draw[white] ($ (cut2)!\wcut!90:(cs) $) -- ($ (cut2)!\wcut!-90:(cs) $);
	
    \node[source, fill=cs0, label={above:Source $1$}] at (src1) {};
    \node[source, fill=cs0, label={above:Source $2$}] at (src2) {};
	
\end{tikzpicture}

%% file: grid_rmse_fs_NSnp.tex
%
\begin{tikzpicture}

\begin{axis}[%
separate axis lines,
every outer x axis line/.append style={black},
every x tick label/.append style={font=\color{black}},
xmode=log,
xmin=1,
xmax=1000,
xminorticks=true,
xlabel={Snapshots},
xmajorgrids,
xminorgrids,
every outer y axis line/.append style={black},
every y tick label/.append style={font=\color{black}},
ymode=log,
ymin=0.001,
ymax=1,
yminorticks=true,
ylabel={RMSE($\hat{\mb{\mu}}$)},
ymajorgrids,
yminorgrids,
axis background/.style={fill=white},
legend style={at={(0.03,0.03)},anchor=south west,legend cell align=left,align=left,draw=black},
height=\ph,
width=\pw,
xlabel near ticks,
ylabel near ticks
]
\addplot [color=\cGBC,solid,mark=\mGBC,mark options={solid}]
  table[row sep=crcr]{%
1	0.346375101106686\\
2	0.178016536685032\\
3	0.13091578336218\\
5	0.0781337848394233\\
7	0.0397086640419947\\
9	0.0438836036198792\\
11	0.0276839062754277\\
14	0.0214174539414313\\
17	0.0184160256298692\\
20	0.0159862545123407\\
30	0.00935107836918644\\
50	0.00725504422224059\\
70	0.00614176956042257\\
100	0.00527910977343717\\
200	0.00418732213552607\\
300	0.00381518894245968\\
500	0.00349804707420202\\
700	0.00334464746522957\\
1000	0.00310386855391782\\
};
\addlegendentry{\mname{} \eqref{eq:sdp1b}};

\addplot [color=\cSR,solid,mark=\mSR,mark options={solid}]
  table[row sep=crcr]{%
1	0.392940847965697\\
2	0.372656139535274\\
3	0.322081391369741\\
5	0.293247771801711\\
7	0.289119396904927\\
9	0.288667804462734\\
11	0.287588214872125\\
14	0.28755095954167\\
17	0.285944970813154\\
20	0.284915963633727\\
30	0.279579983785204\\
50	0.260979880450584\\
70	0.235884127627671\\
100	0.186048710467626\\
200	0.0730549040106134\\
300	0.0206207985619697\\
500	0.00240665743303858\\
700	0.00206768791971448\\
1000	0.00189446562386333\\
};
\addlegendentry{Spect. RARE \cite{Gershman:Rare}};

\addplot [color=black,solid]
  table[row sep=crcr]{%
1	0.0387313554225835\\
2	0.0273872040638552\\
3	0.0223615584793077\\
5	0.0173211887171204\\
7	0.0146390763412299\\
9	0.0129104518075278\\
11	0.0116779430507824\\
14	0.0103513901511912\\
17	0.00939373349592066\\
20	0.00866059435856018\\
30	0.00707134568256634\\
50	0.00547744081277103\\
70	0.00462928240793706\\
100	0.00387313554225835\\
200	0.00273872040638552\\
300	0.00223615584793077\\
500	0.00173211887171204\\
700	0.00146390763412299\\
1000	0.00122479300000877\\
};
\addlegendentry{CRB \cite{Gershman:Rare}};

\end{axis}
\end{tikzpicture}%

%% file: grid_rmse_fs_SNR.tex
%
\begin{tikzpicture}

\begin{axis}[%
separate axis lines,
every outer x axis line/.append style={black},
every x tick label/.append style={font=\color{black}},
xmin=-10,
xmax=30,
xlabel={SNR in dB},
xmajorgrids,
every outer y axis line/.append style={black},
every y tick label/.append style={font=\color{black}},
ymode=log,
ymin=0.0001,
ymax=1,
yminorticks=true,
ylabel={RMSE($\hat{\mb{\mu}}$)},
ymajorgrids,
yminorgrids,
axis background/.style={fill=white},
legend style={at={(0.03,0.03)},anchor=south west,legend cell align=left,align=left,draw=black},
height=\ph,
width=\pw,
xlabel near ticks,
ylabel near ticks
]
\addplot [color=\cGBC,solid,mark=\mGBC,mark options={solid}]
  table[row sep=crcr]{%
-10	0.486354282733622\\
-8	0.299836646081384\\
-6	0.145819240842901\\
-4	0.0633247160804268\\
-2	0.0416445674728409\\
0	0.0254125166010767\\
2	0.0321728819556677\\
4	0.0190270333998761\\
6	0.0176034750357233\\
8	0.0098123561560582\\
10	0.00800431133827263\\
12	0.00670499813571934\\
14	0.00552654201709048\\
16	0.00490353613901911\\
18	0.0041846545058503\\
20	0.00381536367860259\\
22	0.00323202722760807\\
24	0.00301203142967222\\
26	0.00284224089994731\\
28	0.00265260877879368\\
30	0.00276688754138412\\
32	0.00259383885389975\\
34	0.00257429602027427\\
36	0.00243084347501026\\
38	0.00218822606997844\\
40	0.00211770315829829\\
};
\addlegendentry{\mname{} \eqref{eq:sdp1b}};

\addplot [color=\cSR,solid,mark=\mSR,mark options={solid}]
  table[row sep=crcr]{%
-10	0.367748447538079\\
-8	0.340157001985849\\
-6	0.32339902030361\\
-4	0.302272254763817\\
-2	0.290837618497565\\
0	0.28872484536897\\
2	0.287771302715661\\
4	0.287694469880813\\
6	0.284185757911972\\
8	0.279135498160544\\
10	0.267273322649309\\
12	0.231789579144535\\
14	0.17041263744218\\
16	0.102049177687361\\
18	0.0330365857800105\\
20	0.0131434774698329\\
22	0.00209236708060509\\
24	0.00188812075884989\\
26	0.00175679632665066\\
28	0.00166823259769132\\
30	0.00163869867069369\\
32	0.0016152399202595\\
34	0.00160934769394308\\
36	0.00159906222518072\\
38	0.00158965405041472\\
40	0.00156364957711116\\
};
\addlegendentry{Spect. RARE \cite{Gershman:Rare}};

\addplot [color=black,solid]
  table[row sep=crcr]{%
-10	0.0758620205141153\\
-8	0.0545333458913922\\
-6	0.0401685782089806\\
-4	0.0302147741760778\\
-2	0.0231085633517866\\
0	0.0178929054444366\\
2	0.0139753282549441\\
4	0.0109800980595043\\
6	0.00866059435856018\\
8	0.00684847859616731\\
10	0.00542440100612428\\
12	0.00430094643162824\\
14	0.00341244315185495\\
16	0.0027086331952481\\
18	0.00215055745559109\\
20	0.00170775394307175\\
22	0.0013562692395266\\
24	0.00107719866785952\\
26	0.000855587020624918\\
28	0.000679585703861428\\
30	0.000539798462817654\\
32	0.000428769316450378\\
34	0.000340579643048299\\
36	0.000270530056361562\\
38	0.00021488867462172\\
40	0.000170691646656129\\
};
\addlegendentry{CRB \cite{Gershman:Rare}};

\end{axis}
\end{tikzpicture}%

%% file: sep1_rmse_fs.tex
%
\begin{tikzpicture}

\begin{axis}[%
xmode=log,
xmin=0.01,
xmax=1,
xminorticks=true,
xlabel={Frequency Separation $\Delta \mu$},
xmajorgrids,
xminorgrids,
ymode=log,
ymin=0.01,
ymax=10,
yminorticks=true,
ylabel={RMSE($\hat{\mb{\mu}}$)},
ymajorgrids,
yminorgrids,
axis background/.style={fill=white},
legend style={legend cell align=left,align=left,draw=white!15!black},
height=\ph,
width=\pw,
xlabel near ticks,
ylabel near ticks,
every axis plot/.append style={mark indices={\pIdx}},
legend image post style={mark indices={2}}
]
\addplot [color=\cGLC,solid,mark=\mGLC,mark options={solid}]
  table[row sep=crcr]{%
0.01	0.356506065006368\\
0.02	0.354717958122106\\
0.03	0.347229752313816\\
0.04	0.333206426386545\\
0.05	0.300547733643535\\
0.06	0.236599111286444\\
0.07	0.172998451749765\\
0.08	0.112766011879543\\
0.09	0.0638767617879666\\
0.1	0.0442923559328727\\
0.102	0.0403118605576237\\
0.122	0.0291354184839065\\
0.142	0.0263398491921933\\
0.162	0.0247238060042231\\
0.182	0.0241667968571066\\
0.202	0.0233477128876691\\
0.222	0.0234118133408314\\
0.242	0.023843061562398\\
0.262	0.0240122935049674\\
0.282	0.0240763975041883\\
0.302	0.0243793680894757\\
0.322	0.0249759244626783\\
0.342	0.0246481512691059\\
0.362	0.024269688341952\\
0.382	0.0239705738302097\\
0.402	0.0234950524556427\\
0.422	0.0237119005253552\\
0.442	0.0230089966520502\\
0.462	0.0234667808770034\\
0.482	0.0230187470671375\\
0.502	0.0252550905065949\\
0.522	0.0271167554707791\\
0.542	0.0280762934712259\\
0.562	0.0313124933600167\\
0.582	0.0328744917319251\\
0.602	0.0359683973409281\\
0.622	0.0372409901216953\\
0.642	0.0418466858875756\\
0.662	0.0402385631929059\\
0.682	0.0371330607842282\\
0.702	0.032937498264032\\
0.722	0.0314331235990331\\
0.742	0.0293829491060598\\
0.762	0.0281098054913352\\
0.782	0.0259765742044169\\
0.802	0.0249218811247428\\
0.822	0.0235622370986336\\
0.842	0.0240232190152722\\
0.862	0.0228968039932231\\
0.882	0.022832349409098\\
0.902	0.0229526156049913\\
0.922	0.0228352912741003\\
0.942	0.0302628105250387\\
0.962	0.0223416436374013\\
0.982	0.0236384652382118\\
};
\addlegendentry{GL-\mname{} \eqref{eq:smr_dual_gridless}};

\addplot [color=\cGBC,solid,mark=\mGBC,mark options={solid}]
  table[row sep=crcr]{%
0.01	0.35594184918326\\
0.02	0.354656594468507\\
0.03	0.348321977486347\\
0.04	0.33450620323097\\
0.05	0.302241294332855\\
0.06	0.240765030683446\\
0.07	0.17980628465101\\
0.08	0.1177722378152\\
0.09	0.0678203509280216\\
0.1	0.045785368842022\\
0.102	0.0419404339510215\\
0.122	0.029313819266687\\
0.142	0.0264491965851516\\
0.162	0.0248008064385011\\
0.182	0.0243667806654879\\
0.202	0.0235338054721288\\
0.222	0.023004347415217\\
0.242	0.0239737356288084\\
0.262	0.0241337937340982\\
0.282	0.0243228287828533\\
0.302	0.0245132617168748\\
0.322	0.0251053779099219\\
0.342	0.0247701433181158\\
0.362	0.024347484469653\\
0.382	0.0241089195112513\\
0.402	0.023631335129442\\
0.422	0.023734784599823\\
0.442	0.023023466289853\\
0.462	0.0237857099957096\\
0.482	0.023307080469248\\
0.502	0.0253605993620024\\
0.522	0.0271900717174485\\
0.542	0.0282832105674019\\
0.562	0.0313384747554821\\
0.582	0.0329836323045233\\
0.602	0.0362494137883635\\
0.622	0.0396179252359333\\
0.642	0.0399444614433591\\
0.662	0.0403013647411598\\
0.682	0.0372577508714629\\
0.702	0.0329878765609427\\
0.722	0.0316025315441658\\
0.742	0.029548942451465\\
0.762	0.0282095728432742\\
0.782	0.0260879282427715\\
0.802	0.0251817394156956\\
0.822	0.0292017122785634\\
0.842	0.0241805707128677\\
0.862	0.0231032465251098\\
0.882	0.0230768282049331\\
0.902	0.0232193884501724\\
0.922	0.0230377950333793\\
0.942	0.0304932779477707\\
0.962	0.0224071417186574\\
0.982	0.0238121817564035\\
};
\addlegendentry{\mname{} \eqref{eq:sdp1b}};

\addplot [color=\cRR,solid,mark=\mRR,mark options={solid}]
  table[row sep=crcr]{%
0.01	0.492028770531687\\
0.02	0.47935187127136\\
0.03	0.464495635649711\\
0.04	0.41287364906711\\
0.05	0.340551375972727\\
0.06	0.24338761265302\\
0.07	0.176053740239554\\
0.08	0.0989343274221404\\
0.09	0.0689449622706548\\
0.1	0.0573343769416691\\
0.102	0.0429553021521156\\
0.122	0.0313087509204261\\
0.142	0.0277614055603211\\
0.162	0.0257199281363024\\
0.182	0.0245867551070222\\
0.202	0.0235975464045797\\
0.222	0.0229828136077067\\
0.242	0.0241672105361042\\
0.262	0.0240222601547791\\
0.282	0.0242086076136697\\
0.302	0.0246172922277164\\
0.322	0.0250288740088398\\
0.342	0.0246628038267239\\
0.362	0.0247321052463702\\
0.382	0.0236621424763734\\
0.402	0.0237347782741278\\
0.422	0.0238264483255402\\
0.442	0.0230244612322077\\
0.462	0.0237187990133229\\
0.482	0.0230237859838707\\
0.502	0.0248307885394042\\
0.522	0.0265457430998559\\
0.542	0.0270459061859376\\
0.562	0.0293741953277518\\
0.582	0.0300588355817792\\
0.602	0.0330017131123798\\
0.622	0.0339424499183105\\
0.642	0.0343899522118419\\
0.662	0.0333531860164383\\
0.682	0.0315465447942318\\
0.702	0.0304741340958495\\
0.722	0.029088772255821\\
0.742	0.0272024264199021\\
0.762	0.0263539750429463\\
0.782	0.0252449025911619\\
0.802	0.024575585134711\\
0.822	0.02317528394408\\
0.842	0.0240871153492705\\
0.862	0.0231172446240696\\
0.882	0.0228882225639955\\
0.902	0.0233259385095857\\
0.922	0.0232209694313086\\
0.942	0.0232477101244517\\
0.962	0.0227695812482126\\
0.982	0.0236622395105868\\
};
\addlegendentry{Root-RARE \cite{Pesavento:Rare}};

\addplot [color=\cSR,solid,mark=\mSR,mark options={solid}]
  table[row sep=crcr]{%
0.01	0.539119374535919\\
0.02	0.543295039550335\\
0.03	0.558707526349878\\
0.04	0.580754423142863\\
0.05	0.589498854960719\\
0.06	0.620244145478213\\
0.07	0.626041612035493\\
0.08	0.645571297379305\\
0.09	0.652141242370086\\
0.1	0.654145931730833\\
0.102	0.657631081990502\\
0.122	0.643630484051212\\
0.142	0.609516513311983\\
0.162	0.549284407206319\\
0.182	0.472017965759777\\
0.202	0.400505255895599\\
0.222	0.305066943473068\\
0.242	0.249821656387112\\
0.262	0.188521988107489\\
0.282	0.155257334770374\\
0.302	0.135138077535534\\
0.322	0.0899103998433997\\
0.342	0.07676887390082\\
0.362	0.0590586149515886\\
0.382	0.029745924090537\\
0.402	0.026767891213168\\
0.422	0.0307788888688334\\
0.442	0.030750609750052\\
0.462	0.0255150935722367\\
0.482	0.0245206035814781\\
0.502	0.0263784002547538\\
0.522	0.0304447039729407\\
0.542	0.0287899287946322\\
0.562	0.0311107698393979\\
0.582	0.0315455226617026\\
0.602	0.0341727961981456\\
0.622	0.0374093571182398\\
0.642	0.0414955419292242\\
0.662	0.0364433807432845\\
0.682	0.0376592618090158\\
0.702	0.0310766793592879\\
0.722	0.0295827652527615\\
0.742	0.0277812886670146\\
0.762	0.02660864521166\\
0.782	0.0255080379488505\\
0.802	0.0248632258566743\\
0.822	0.0234525051966735\\
0.842	0.0242165232847326\\
0.862	0.0231892216341989\\
0.882	0.0228923568030904\\
0.902	0.0233486616318794\\
0.922	0.0231473540604536\\
0.942	0.0232379000772445\\
0.962	0.0227398328929656\\
0.982	0.023631335129442\\
};
\addlegendentry{Spect. RARE \cite{Gershman:Rare}};

\addplot [color=black,solid]
  table[row sep=crcr]{%
0.01	1.12378591908282\\
0.02	0.305552262560784\\
0.03	0.15257891839586\\
0.04	0.0978179204970321\\
0.05	0.0715642735845381\\
0.06	0.0566632477821595\\
0.07	0.0472370524400105\\
0.08	0.0408183254895679\\
0.09	0.0362156098176964\\
0.1	0.0327911070834543\\
0.102	0.0322120715728824\\
0.122	0.0277894427205256\\
0.142	0.0250520640551642\\
0.162	0.0233444865102264\\
0.182	0.0223289007700433\\
0.202	0.0218098930308724\\
0.222	0.0216582607881478\\
0.242	0.021771812004125\\
0.262	0.0220529444941502\\
0.282	0.022398195204496\\
0.302	0.0227014580372359\\
0.322	0.0228719670084662\\
0.342	0.0228598210106035\\
0.362	0.0226726855489986\\
0.382	0.0223705683094504\\
0.402	0.0220422171746323\\
0.422	0.0217791408896905\\
0.442	0.0216598182981513\\
0.462	0.0217457267185823\\
0.482	0.0220841177743635\\
0.502	0.0227116877077661\\
0.522	0.0236544229191102\\
0.542	0.0249189696945221\\
0.562	0.0264698261214996\\
0.582	0.0281886735935699\\
0.602	0.0298289247062857\\
0.622	0.0310209464593188\\
0.642	0.0314074849172878\\
0.662	0.0308703861827731\\
0.682	0.0296193449455777\\
0.702	0.0280311235820424\\
0.722	0.026435615158542\\
0.742	0.0250283357037431\\
0.762	0.0238891372172874\\
0.782	0.0230281961725229\\
0.802	0.0224209295122026\\
0.822	0.0220272285473387\\
0.842	0.0218006564362637\\
0.862	0.021693308210346\\
0.882	0.0216598683469754\\
0.902	0.0216621022461978\\
0.922	0.0216730045062792\\
0.942	0.0216787440616036\\
0.962	0.0216770563019443\\
0.982	0.0216724945532767\\
};
\addlegendentry{CRB \cite{Gershman:Rare}};

\end{axis}
\end{tikzpicture}%

%% file: sep2_rmse_fs.tex
%
\begin{tikzpicture}

\begin{axis}[%
xmode=log,
xmin=0.01,
xmax=1,
xminorticks=true,
xlabel={Frequency Separation $\Delta \mu$},
xmajorgrids,
xminorgrids,
ymode=log,
ymin=0.001,
ymax=1,
yminorticks=true,
ylabel={RMSE($\hat{\mb{\mu}}$)},
ymajorgrids,
yminorgrids,
axis background/.style={fill=white},
legend style={legend cell align=left,align=left,draw=white!15!black},
height=\ph,
width=\pw,
xlabel near ticks,
ylabel near ticks,
every axis plot/.append style={mark indices={\pIdx}},
legend image post style={mark indices={2}}
]
\addplot [color=\cGLC,solid,mark=\mGLC,mark options={solid}]
  table[row sep=crcr]{%
0.01	0.0202564416361183\\
0.02	0.0124380416910272\\
0.03	0.0116858913457593\\
0.04	0.0114351875777565\\
0.05	0.0115186027509737\\
0.06	0.0123174528716432\\
0.07	0.0126733696952275\\
0.08	0.0129377675271596\\
0.09	0.0123855647686024\\
0.1	0.0113705837344042\\
0.102	0.0108211633089937\\
0.122	0.00771874699790496\\
0.142	0.00494906766128581\\
0.162	0.00287062150049031\\
0.182	0.00191073329822997\\
0.202	0.00148057039516016\\
0.222	0.00139765817326799\\
0.242	0.00143528652857972\\
0.262	0.00164447654994532\\
0.282	0.00191589674713256\\
0.302	0.00227349550968045\\
0.322	0.00234015525409711\\
0.342	0.00235524903313086\\
0.362	0.00205642618705356\\
0.382	0.00179079345918827\\
0.402	0.0015142939126232\\
0.422	0.00133322220592488\\
0.442	0.00130958686461354\\
0.462	0.00135567281978868\\
0.482	0.00150080733240625\\
0.502	0.00173831886353228\\
0.522	0.00217035936645098\\
0.542	0.00255281569416292\\
0.562	0.00304868417931985\\
0.582	0.0033806660176098\\
0.602	0.00367677352082896\\
0.622	0.00364038841816138\\
0.642	0.00325740062202138\\
0.662	0.00290714314572022\\
0.682	0.0027464737403563\\
0.702	0.0026680617086402\\
0.722	0.00237893314433731\\
0.742	0.00218654242464424\\
0.762	0.00201162968639872\\
0.782	0.00172894170041086\\
0.802	0.00154546360869986\\
0.822	0.00142526478734946\\
0.842	0.00133748611112502\\
0.862	0.00134758997620425\\
0.882	0.00134816087355342\\
0.902	0.00131702126593698\\
0.922	0.00132945370759797\\
0.942	0.00133238229492842\\
0.962	0.00132551638609404\\
0.982	0.00132570700179244\\
};
\addlegendentry{GL-\mname{} \eqref{eq:smr_dual_gridless}};

\addplot [color=\cGBC,solid,mark=\mGBC,mark options={solid}]
  table[row sep=crcr]{%
0.01	0.0950865901551279\\
0.02	0.0124756491593092\\
0.03	0.011753595106588\\
0.04	0.0117931622115632\\
0.05	0.0117887724246211\\
0.06	0.0123756757349441\\
0.07	0.012823454830459\\
0.08	0.0129639570423008\\
0.09	0.0124174291455137\\
0.1	0.0112126998390736\\
0.102	0.0107321038864814\\
0.122	0.00761264612026463\\
0.142	0.00522396928326184\\
0.162	0.00427593777928764\\
0.182	0.00415810715019029\\
0.202	0.00420760427647514\\
0.222	0.00424678663077988\\
0.242	0.00421497871271663\\
0.262	0.00419281649311386\\
0.282	0.00417301793369601\\
0.302	0.00416059599170738\\
0.322	0.00414064337225898\\
0.342	0.00417053650519944\\
0.362	0.00416308334530537\\
0.382	0.004195284743456\\
0.402	0.00421006385713557\\
0.422	0.00421006385713557\\
0.442	0.00423458109128086\\
0.462	0.00424922353155469\\
0.482	0.00421743399265498\\
0.502	0.00419528474345604\\
0.522	0.00418787562821925\\
0.542	0.00420760427647517\\
0.562	0.00422968901352556\\
0.582	0.00446309922121223\\
0.602	0.00449775651459663\\
0.622	0.00469151915341649\\
0.642	0.00466940174212927\\
0.662	0.00464272176900909\\
0.682	0.00474854232836842\\
0.702	0.00488607265645332\\
0.722	0.00486909381036775\\
0.742	0.00484565075742567\\
0.762	0.00475507792213216\\
0.782	0.0046560808655902\\
0.802	0.00451383885410504\\
0.822	0.00435984880167868\\
0.842	0.00432408632745554\\
0.862	0.00428802560172823\\
0.882	0.00431209959283768\\
0.902	0.00427835807595473\\
0.922	0.00429043907949108\\
0.942	0.00429526196666966\\
0.962	0.00432647968919714\\
0.982	0.00427351611189103\\
};
\addlegendentry{\mname{} \eqref{eq:sdp1b}};

\addplot [color=\cRR,solid,mark=\mRR,mark options={solid}]
  table[row sep=crcr]{%
0.01	0.0178412721961968\\
0.02	0.00820673789832808\\
0.03	0.00600927917821999\\
0.04	0.00427534141553072\\
0.05	0.00350690114702852\\
0.06	0.00298459052844674\\
0.07	0.00266768279321611\\
0.08	0.00231347158743287\\
0.09	0.00208444739833632\\
0.1	0.00189460492753316\\
0.102	0.00188134273540603\\
0.122	0.00162638873566038\\
0.142	0.00153335105205049\\
0.162	0.00144617498244105\\
0.182	0.00134945596170222\\
0.202	0.00133182069032217\\
0.222	0.00132344023626281\\
0.242	0.00134560620334995\\
0.262	0.00135641885353335\\
0.282	0.00137874658621245\\
0.302	0.00135534230484483\\
0.322	0.00141935933054015\\
0.342	0.00143236374732427\\
0.362	0.00138462683192656\\
0.382	0.00142546776996483\\
0.402	0.00133177682807331\\
0.422	0.00128549475835592\\
0.442	0.00129990470179095\\
0.462	0.00133120354772938\\
0.482	0.00138205047563981\\
0.502	0.00141127359344768\\
0.522	0.00146557128880431\\
0.542	0.0015108688045405\\
0.562	0.00162247994516594\\
0.582	0.00171453544493837\\
0.602	0.00185526920547059\\
0.622	0.00192983108886828\\
0.642	0.0019286075594714\\
0.662	0.00187447014579412\\
0.682	0.0017887820527742\\
0.702	0.0017059769876947\\
0.722	0.0016061677250658\\
0.742	0.00152865048957633\\
0.762	0.00150382479242184\\
0.782	0.00139189237617936\\
0.802	0.00137842290022286\\
0.822	0.00134364266696135\\
0.842	0.00131059203408289\\
0.862	0.00134006562203264\\
0.882	0.00135529945315676\\
0.902	0.00132532388002251\\
0.922	0.00133728541698145\\
0.942	0.00133085014358447\\
0.962	0.00133262010300829\\
0.982	0.00133070948346191\\
};
\addlegendentry{Root-RARE \cite{Pesavento:Rare}};

\addplot [color=\cSR,solid,mark=\mSR,mark options={solid}]
  table[row sep=crcr]{%
0.01	0.692653963720383\\
0.02	0.697925089563218\\
0.03	0.672237888301389\\
0.04	0.472626411278173\\
0.05	0.147056523547962\\
0.06	0.00510246966550225\\
0.07	0.00502066124504275\\
0.08	0.00502066124504275\\
0.09	0.00500000000000001\\
0.1	0.005\\
0.102	0.00462261073644582\\
0.122	0.00458212387580811\\
0.142	0.00444450483050763\\
0.162	0.00439296444852022\\
0.182	0.00429526196666963\\
0.202	0.0043240863274555\\
0.222	0.00432647968919713\\
0.242	0.00432408632745554\\
0.262	0.00436696615223216\\
0.282	0.00432647968919715\\
0.302	0.00428319456639341\\
0.322	0.00432647968919715\\
0.342	0.00430007944459331\\
0.362	0.00430248616076349\\
0.382	0.00432169164026542\\
0.402	0.00431209959283764\\
0.422	0.00425895719115695\\
0.442	0.00429526196666963\\
0.462	0.00429285120037622\\
0.482	0.00428077700421717\\
0.502	0.00432169164026548\\
0.522	0.00431689828071564\\
0.542	0.00433842669291212\\
0.562	0.00436222254215056\\
0.582	0.00441412188885735\\
0.602	0.00444217505898389\\
0.622	0.00447699433335243\\
0.642	0.00447236742602561\\
0.662	0.00449315098757004\\
0.682	0.00439767489856376\\
0.702	0.00438824894217136\\
0.722	0.00439060732840179\\
0.742	0.00436696615223216\\
0.762	0.00431689828071565\\
0.782	0.00430969824521433\\
0.802	0.00429043907949107\\
0.822	0.00425652586299676\\
0.842	0.00427351611189103\\
0.862	0.00427109307143322\\
0.882	0.00429526196666967\\
0.902	0.00428077700421717\\
0.922	0.0042759377792877\\
0.942	0.00427109307143321\\
0.962	0.0043000794445933\\
0.982	0.00426866865557605\\
};
\addlegendentry{Spect. RARE \cite{Gershman:Rare}};

\addplot [color=black,solid]
  table[row sep=crcr]{%
0.01	0.0181556339339397\\
0.02	0.00860136271583879\\
0.03	0.00570166226337118\\
0.04	0.00429298334261814\\
0.05	0.00346314667816751\\
0.06	0.0029192438067877\\
0.07	0.00253785719042494\\
0.08	0.00225786789942506\\
0.09	0.00204553855126857\\
0.1	0.00188072856388398\\
0.102	0.00185225722346674\\
0.122	0.00162878275714817\\
0.142	0.00148494065296306\\
0.162	0.00139296281337737\\
0.182	0.00133737429385545\\
0.202	0.00130866360475139\\
0.222	0.00130021732207947\\
0.242	0.00130659234767793\\
0.262	0.0013224552045677\\
0.282	0.00134206691879465\\
0.302	0.00135946619828307\\
0.322	0.00136946314874005\\
0.342	0.00136909517845061\\
0.362	0.00135864466204474\\
0.382	0.00134143361365914\\
0.402	0.00132253220733044\\
0.422	0.00130727198344563\\
0.442	0.00130030862076393\\
0.462	0.00130536121573233\\
0.482	0.0013253507860297\\
0.502	0.00136259669145583\\
0.522	0.00141878597595559\\
0.542	0.00149443257552114\\
0.562	0.00158748396256044\\
0.582	0.00169085852797845\\
0.602	0.00178970526698303\\
0.622	0.00186172136644004\\
0.642	0.00188530891991202\\
0.662	0.00185324528755004\\
0.682	0.00177808270218681\\
0.702	0.00168250823039172\\
0.722	0.00158644210309159\\
0.742	0.00150172548768283\\
0.762	0.00143321818689379\\
0.782	0.00138155463405702\\
0.802	0.00134524355452874\\
0.822	0.00132183151759681\\
0.842	0.00130846863838553\\
0.862	0.0013022164523333\\
0.882	0.00130030683643691\\
0.902	0.00130042297769068\\
0.922	0.00130095355852283\\
0.942	0.00130110739897171\\
0.962	0.00130080717691586\\
0.982	0.00130038801205081\\
};
\addlegendentry{CRB \cite{Gershman:Rare}};

\end{axis}
\end{tikzpicture}%

%% file: sep2_bias_fs.tex
%
\begin{tikzpicture}

\begin{axis}[%
xmode=log,
xmin=0.01,
xmax=1,
xminorticks=true,
xlabel={Frequency Separation $\Delta \mu$},
xmajorgrids,
xminorgrids,
ymode=log,
ymin=0.0001,
ymax=1,
yminorticks=true,
ylabel={Bias($\hat{\mb{\mu}}$)},
ymajorgrids,
yminorgrids,
axis background/.style={fill=white},
legend style={legend cell align=left,align=left,draw=white!15!black},
height=\ph,
width=\pw,
xlabel near ticks,
ylabel near ticks,
every axis plot/.append style={mark indices={\pIdx}},
legend image post style={mark indices={2}}
]
\addplot [color=\cGLC,solid,mark=\mGLC,mark options={solid}]
  table[row sep=crcr]{%
0.01	0.0172353409128593\\
0.02	0.0144178477164969\\
0.03	0.0149359028252515\\
0.04	0.0145981844564355\\
0.05	0.0140358595315109\\
0.06	0.0130917475970332\\
0.07	0.0126549318587544\\
0.08	0.0115276526490279\\
0.09	0.0114488344968686\\
0.1	0.0103732021300909\\
0.102	0.00997309407076371\\
0.122	0.00772614258659186\\
0.142	0.00502240127536571\\
0.162	0.00293522956392042\\
0.182	0.00159498273387897\\
0.202	0.000732941749213827\\
0.222	0.000473634998882355\\
0.242	0.000572724040157506\\
0.262	0.00103616016549632\\
0.282	0.00155771795791522\\
0.302	0.00225736644841469\\
0.322	0.00228395245874145\\
0.342	0.0022744698189778\\
0.362	0.00184289376049261\\
0.382	0.00122679098929032\\
0.402	0.00074746480155856\\
0.422	0.000299575334503186\\
0.442	0.000193405018176496\\
0.462	0.000257242262959779\\
0.482	0.000602013233045809\\
0.502	0.00107366115219982\\
0.522	0.00182212257699588\\
0.542	0.00237183483747151\\
0.562	0.0029088173890233\\
0.582	0.00310982499329145\\
0.602	0.00336621506073973\\
0.622	0.00258475687350275\\
0.642	0.00141308429893914\\
0.662	0.000409396687913082\\
0.682	0.000704320252145031\\
0.702	0.00145993605070441\\
0.722	0.00156212249343028\\
0.742	0.00161261479172211\\
0.762	0.0013694476920227\\
0.782	0.00111358476082574\\
0.802	0.000774523808906985\\
0.822	0.000395680859615153\\
0.842	0.000282189023347794\\
0.862	7.16811359522082e-05\\
0.882	6.24345077548575e-05\\
0.902	5.86271894410536e-05\\
0.922	7.70521149950439e-05\\
0.942	0.000160380775276912\\
0.962	0.000174024498675749\\
0.982	4.08723092857903e-05\\
};
\addlegendentry{GL-\mname{} \eqref{eq:smr_dual_gridless}};

\addplot [color=\cGBC,solid,mark=\mGBC,mark options={solid}]
  table[row sep=crcr]{%
0.01	0.11356922813203\\
0.02	0.0136822028800699\\
0.03	0.014319977215294\\
0.04	0.0143691005608742\\
0.05	0.0137761383249892\\
0.06	0.0128025876029164\\
0.07	0.0124515917286654\\
0.08	0.0111777085036802\\
0.09	0.011253349840217\\
0.1	0.00996081938567006\\
0.102	0.00957800785916963\\
0.122	0.00746079718498596\\
0.142	0.00583903092981868\\
0.162	0.00492638172184835\\
0.182	0.00384651341635194\\
0.202	0.0030541397183231\\
0.222	0.00270994325834597\\
0.242	0.00276761417070637\\
0.262	0.00334865311784339\\
0.282	0.00371750992775758\\
0.302	0.00451046541654353\\
0.322	0.00442085760096704\\
0.342	0.00430306546504939\\
0.362	0.00417406965512313\\
0.382	0.00340040666639958\\
0.402	0.00296115185115126\\
0.422	0.0027273462532399\\
0.442	0.00256611958327738\\
0.462	0.00252779769710507\\
0.482	0.00280849910596626\\
0.502	0.00328628304021651\\
0.522	0.00373398243341157\\
0.542	0.00408388228922937\\
0.562	0.00419487500203121\\
0.582	0.00399719240999978\\
0.602	0.00407577641855415\\
0.622	0.00302822617348943\\
0.642	0.00199854204882988\\
0.662	0.000924462520232017\\
0.682	0.000844002150609794\\
0.702	0.0017831196080185\\
0.722	0.00195774387541518\\
0.742	0.00238273495247376\\
0.762	0.00217229654973626\\
0.782	0.00196718532999791\\
0.802	0.00188733248989056\\
0.822	0.00211635930783458\\
0.842	0.00221746862860951\\
0.862	0.00230755628333331\\
0.882	0.00220377401305661\\
0.902	0.00236694140212191\\
0.922	0.0023028005572227\\
0.942	0.00227557424641072\\
0.962	0.00215286667127982\\
0.982	0.00237756172290721\\
};
\addlegendentry{\mname{} \eqref{eq:sdp1b}};

\addplot [color=\cRR,solid,mark=\mRR,mark options={solid}]
  table[row sep=crcr]{%
0.01	0.00929168009691691\\
0.02	0.00255442493482925\\
0.03	0.00140424602512256\\
0.04	0.000425295664559765\\
0.05	0.000127195601978606\\
0.06	6.10564182939711e-05\\
0.07	5.73107950374557e-05\\
0.08	8.33821027392308e-05\\
0.09	5.61072544582201e-05\\
0.10	7.83401838047987e-05\\
0.102	0.0000110007231037238\\
0.122	9.75094847168142e-05\\
0.142	1.04537327943169e-05\\
0.162	3.81656248923427e-05\\
0.182	7.77539642629847e-05\\
0.202	7.54858453929572e-05\\
0.222	9.41077515726066e-05\\
0.242	8.6545683632014e-05\\
0.262	6.63677705202529e-05\\
0.282	8.99546894487635e-05\\
0.302	8.57042295491467e-05\\
0.322	2.07090693131564e-05\\
0.342	5.33256589338571e-05\\
0.362	9.19847748484087e-06\\
0.382	6.62850686285967e-05\\
0.402	6.02526903219936e-05\\
0.422	6.97583495829249e-05\\
0.442	2.90287788298924e-05\\
0.462	2.66820342165528e-05\\
0.482	2.97725889365217e-05\\
0.502	6.99395339095627e-05\\
0.522	0.000012226644875766\\
0.542	4.59652386471722e-05\\
0.562	6.4431712367186e-05\\
0.582	4.64567417605621e-05\\
0.602	6.24340005800141e-05\\
0.622	8.74583178765688e-05\\
0.642	4.96859102285359e-05\\
0.662	8.8328202403562e-05\\
0.682	8.86783079008762e-05\\
0.702	7.7036518372571e-05\\
0.722	0.0000140841073444554\\
0.742	2.62567135125109e-05\\
0.762	5.1334609727473e-05\\
0.782	7.58183626694558e-05\\
0.802	3.76645569096379e-05\\
0.822	6.36651541659164e-05\\
0.842	7.12753340777125e-05\\
0.862	3.9126096162057e-05\\
0.882	2.23505898365639e-05\\
0.902	4.73059520150923e-05\\
0.922	2.01272173958555e-05\\
0.942	7.93912966969285e-05\\
0.962	9.06024407939994e-05\\
0.982	6.33678627279373e-05\\
};
\addlegendentry{Root-RARE \cite{Pesavento:Rare}};

\addplot [color=\cSR,solid,mark=\mSR,mark options={solid}]
  table[row sep=crcr]{%
0.01	0.995157697349714\\
0.02	0.988942482031738\\
0.03	0.917569726894366\\
0.04	0.4539417586069\\
0.05	0.041532097547942\\
0.06	0.00308483402069921\\
0.07	0.00242379759518065\\
0.08	0.00237953418921138\\
0.09	0.00186901027877581\\
0.1	0.00163790231798775\\
0.102	0.00133429745371636\\
0.122	0.00143333589029248\\
0.142	0.00184143538252213\\
0.162	0.00199536517858647\\
0.182	0.00231849487218743\\
0.202	0.00218869665138734\\
0.222	0.00214800793220571\\
0.242	0.00225703361855048\\
0.262	0.00203449836267648\\
0.282	0.00220742774483382\\
0.302	0.00232942349901425\\
0.322	0.00214980320692666\\
0.342	0.00228196518677514\\
0.362	0.00227336452006789\\
0.382	0.00218223881373098\\
0.402	0.00222766630359368\\
0.422	0.00244786348263581\\
0.442	0.00228778561591802\\
0.462	0.0022909826537974\\
0.482	0.00234262608756488\\
0.502	0.00217692879650046\\
0.522	0.00218280838069163\\
0.542	0.00208905256403198\\
0.562	0.0019962726047997\\
0.582	0.00175925699942965\\
0.602	0.00163406510293793\\
0.622	0.00148995841591967\\
0.642	0.00150698613319738\\
0.662	0.00140960342375399\\
0.682	0.00183095949509847\\
0.702	0.00189555905367229\\
0.722	0.00188215270077779\\
0.742	0.00198606848329687\\
0.762	0.00218280838069534\\
0.782	0.00221346839204808\\
0.802	0.00229615959575221\\
0.822	0.00244364838437499\\
0.842	0.00237079117190188\\
0.862	0.00238538250200409\\
0.882	0.00227545651130888\\
0.902	0.00234573467890755\\
0.922	0.00235854157162538\\
0.942	0.00238113330844818\\
0.962	0.00225867107328592\\
0.982	0.00240106814386908\\
};
\addlegendentry{Spect. RARE \cite{Gershman:Rare}};

\end{axis}
\end{tikzpicture}%

%% file: corr_rmse_fs.tex
%
\begin{tikzpicture}

\begin{axis}[%
separate axis lines,
every outer x axis line/.append style={black},
every x tick label/.append style={font=\color{black}},
xmin=0,
xmax=1,
xlabel={Correlation Coefficient $\rho$},
xmajorgrids,
every outer y axis line/.append style={black},
every y tick label/.append style={font=\color{black}},
ymode=log,
ymin=0.01,
ymax=1,
yminorticks=true,
ylabel={RMSE($\hat{\mb{\mu}}$)},
ymajorgrids,
yminorgrids,
axis background/.style={fill=white},
legend style={at={(0.03,0.97)},anchor=north west,legend cell align=left,align=left,draw=black},
height=\ph,
width=\pw,
xlabel near ticks,
ylabel near ticks
]
\addplot [color=\cGLC,solid,mark=\mGLC,mark options={solid}]
  table[row sep=crcr]{%
0	0.0158765779278305\\
0.05	0.0163409113038866\\
0.1	0.0185268978091396\\
0.15	0.0163362259664747\\
0.2	0.0167783720334671\\
0.25	0.0167388955057332\\
0.3	0.016832649284228\\
0.35	0.0166486063427323\\
0.4	0.0166996001539399\\
0.45	0.0174589278821114\\
0.5	0.0180062505665374\\
0.55	0.020370288239026\\
0.6	0.0182635659953159\\
0.65	0.020630894958835\\
0.7	0.0205929318487916\\
0.75	0.0206380461478615\\
0.8	0.0242863256495007\\
0.85	0.0300469180201519\\
0.9	0.0260367521051778\\
0.95	0.0303930649672286\\
1	0.0377814826847962\\
};
\addlegendentry{GL-\mname{} \eqref{eq:smr_dual_gridless}};

\addplot [color=\cGBC,solid,mark=\mGBC,mark options={solid}]
  table[row sep=crcr]{%
0	0.0162665300540711\\
0.05	0.0166282891483159\\
0.1	0.0165801085641801\\
0.15	0.0165861387911713\\
0.2	0.0170557908054713\\
0.25	0.0170088212407562\\
0.3	0.0168789810118976\\
0.35	0.0168937858397696\\
0.4	0.0170528589978337\\
0.45	0.0177172232587389\\
0.5	0.0182208671582886\\
0.55	0.0188015956769632\\
0.6	0.0184634774622767\\
0.65	0.0193002590656188\\
0.7	0.020532900428337\\
0.75	0.0206784912409006\\
0.8	0.0230477764654207\\
0.85	0.025707975416201\\
0.9	0.025121703763877\\
0.95	0.0288721318921896\\
1	0.0347893661914096\\
};
\addlegendentry{\mname{} \eqref{eq:sdp1b}};

\addplot [color=\cRR,solid,mark=\mRR,mark options={solid}]
  table[row sep=crcr]{%
0	0.0154251194699526\\
0.05	0.015904366890871\\
0.1	0.0157560498040972\\
0.15	0.0158023680024176\\
0.2	0.0165039190907536\\
0.25	0.01703489498088\\
0.3	0.0168888701614747\\
0.35	0.0175821366375048\\
0.4	0.0177336870884281\\
0.45	0.0189140067007765\\
0.5	0.0197930275091955\\
0.55	0.0210869909379714\\
0.6	0.0353619687785822\\
0.65	0.0336573549622547\\
0.7	0.0344120437508326\\
0.75	0.0770420754424223\\
0.8	0.108726549466196\\
0.85	0.174653078125931\\
0.9	0.243934783496059\\
0.95	0.331986528069339\\
1	0.412276200378482\\
};
\addlegendentry{Root-RARE \cite{Pesavento:Rare}};

\addplot [color=\cSR,solid,mark=\mSR,mark options={solid}]
  table[row sep=crcr]{%
0	0.0158619040471186\\
0.05	0.0164742222881688\\
0.1	0.0163340135912763\\
0.15	0.0162234398325385\\
0.2	0.017026450011673\\
0.25	0.0176720117700278\\
0.3	0.0175527775579821\\
0.35	0.0182537667345674\\
0.4	0.018409236811992\\
0.45	0.0198292712927127\\
0.5	0.0205523721258642\\
0.55	0.0260653026071058\\
0.6	0.0315404502187272\\
0.65	0.0247244817943672\\
0.7	0.0288236014404861\\
0.75	0.0672673769371157\\
0.8	0.0998308569531488\\
0.85	0.158586884703622\\
0.9	0.224648837076892\\
0.95	0.294874210469481\\
1	0.36609657195882\\
};
\addlegendentry{Spect. RARE \cite{Gershman:Rare}};

\addplot [color=black,solid]
  table[row sep=crcr]{%
0	0.0148510003025087\\
0.05	0.0148725620066693\\
0.1	0.0149160864581528\\
0.15	0.0149821709520228\\
0.2	0.0150717809960178\\
0.25	0.0151862926815827\\
0.3	0.0153275562369434\\
0.35	0.0154979860842071\\
0.4	0.0157006850409395\\
0.45	0.0159396130397703\\
0.5	0.0162198133721096\\
0.55	0.0165477099898134\\
0.6	0.0169314814460579\\
0.65	0.017381481517775\\
0.7	0.0179105569212768\\
0.75	0.01853374088774\\
0.8	0.019265676424719\\
0.85	0.0201107833441235\\
0.9	0.0210318833894232\\
0.95	0.0218629482790746\\
1	0.0221306405552721\\
};
\addlegendentry{CRB \cite{Gershman:Rare}};

\end{axis}
\end{tikzpicture}%

%% file: cal_rmse_fs.tex
%
\begin{tikzpicture}

\begin{axis}[%
separate axis lines,
every outer x axis line/.append style={black},
every x tick label/.append style={font=\color{black}},
xmin=-10,
xmax=30,
xlabel={SNR in dB},
xmajorgrids,
every outer y axis line/.append style={black},
every y tick label/.append style={font=\color{black}},
ymode=log,
ymin=0.0001,
ymax=1,
yminorticks=true,
ylabel={RMSE($\hat{\mb{\mu}}$)},
ymajorgrids,
yminorgrids,
axis background/.style={fill=white},
legend style={legend cell align=left,align=left,draw=black},
height=\ph,
width=\pw,
xlabel style={yshift=-1mm},
ylabel near ticks
]
\addplot [color=\cGLC,solid,mark=\mGLC,mark options={solid}]
  table[row sep=crcr]{%
-10	0.644152028507248\\
-8	0.400944843764341\\
-6	0.182846943841612\\
-4	0.0787575108425365\\
-2	0.0484937879553249\\
0	0.0334327608130882\\
2	0.0191577152874469\\
4	0.0133459257266029\\
6	0.0109256032072714\\
8	0.00937014407247238\\
10	0.00787285642090324\\
12	0.00669470341526383\\
14	0.00614646482446654\\
16	0.00558334693345304\\
18	0.0050899061213607\\
20	0.00520046735487674\\
22	0.00499250137409232\\
24	0.00471042804823331\\
26	0.00497949800360728\\
28	0.00499407601893865\\
30	0.00478731487478086\\
};
\addlegendentry{GL-\mname{} \eqref{eq:smr_dual_gridless}};

\addplot [color=\cGBC,solid,mark=\mGBC,mark options={solid}]
  table[row sep=crcr]{%
-10	0.670621453081641\\
-8	0.420010952238152\\
-6	0.189146856525117\\
-4	0.0800424887169309\\
-2	0.0485551233136113\\
0	0.0336481797427437\\
2	0.0195004273457446\\
4	0.0136161668614922\\
6	0.0112664694262814\\
8	0.00976387901058457\\
10	0.00835463942968219\\
12	0.00734846922834956\\
14	0.00681664629173812\\
16	0.00644980619863887\\
18	0.00609371260672291\\
20	0.00639270417481264\\
22	0.00634034699365897\\
24	0.00620215016479503\\
26	0.00646529195009788\\
28	0.00640832791503892\\
30	0.00631400559602754\\
};
\addlegendentry{\mname{} \eqref{eq:sdp1b}};

\addplot [color=\cRR,solid,mark=\mRR,mark options={solid}]
  table[row sep=crcr]{%
-10	0.326687264939469\\
-8	0.291475641687875\\
-6	0.268000560331744\\
-4	0.218184160852268\\
-2	0.173744120970786\\
0	0.101171832904401\\
2	0.0464071111697444\\
4	0.0336101269085998\\
6	0.00959504379833237\\
8	0.00770309523613234\\
10	0.0061277964289148\\
12	0.00470821704363093\\
14	0.003882197688264\\
16	0.00294835298512833\\
18	0.00236682310433464\\
20	0.00191535047489859\\
22	0.00150603606930537\\
24	0.00118331017323473\\
26	0.000938157726928196\\
28	0.000749195011577636\\
30	0.000597274937836619\\
};
\addlegendentry{Root-RARE \cite{Pesavento:Rare}};

\addplot [color=\cSR,solid,mark=\mSR,mark options={solid}]
  table[row sep=crcr]{%
-10	0.327125154438888\\
-8	0.318912422670133\\
-6	0.293376663466382\\
-4	0.221347840890005\\
-2	0.148803673789774\\
0	0.08417838202294\\
2	0.0388501394248893\\
4	0.0206025888340925\\
6	0.0101882285015601\\
8	0.00833066623986344\\
10	0.00675771164423779\\
12	0.00561545486433055\\
14	0.00521536192416215\\
16	0.00501996015920448\\
18	0.00500000000000003\\
20	0.00500000000000002\\
22	0.00500000000000003\\
24	0.00500000000000003\\
26	0.00500000000000003\\
28	0.00500000000000002\\
30	0.00500000000000003\\
};
\addlegendentry{Spect. RARE \cite{Gershman:Rare}};

\addplot [color=black,solid]
  table[row sep=crcr]{%
-10	0.0825090033447158\\
-8	0.0584628718083717\\
-6	0.0424677672398413\\
-4	0.0315559143022182\\
-2	0.0238973021932674\\
0	0.0183668912275962\\
2	0.0142700190827908\\
4	0.0111712346332294\\
6	0.00879022324628805\\
8	0.00694009605724767\\
10	0.005491412388636\\
12	0.00435126307025011\\
14	0.00345094368857716\\
16	0.00273847735069354\\
18	0.00217389289453999\\
20	0.00172610394324339\\
22	0.0013707518638668\\
24	0.00108865583999589\\
26	0.000864664319816766\\
28	0.000686784300482332\\
30	0.000545510617398482\\
};
\addlegendentry{CRB \cite{Gershman:Rare}};

\end{axis}
\end{tikzpicture}%

%% file: cal_rmse_phi.tex
%
\begin{tikzpicture}

\begin{axis}[%
separate axis lines,
every outer x axis line/.append style={black},
every x tick label/.append style={font=\color{black}},
xmin=-10,
xmax=30,
xlabel={SNR in dB},
xmajorgrids,
every outer y axis line/.append style={black},
every y tick label/.append style={font=\color{black}},
ymode=log,
ymin=0.001,
ymax=10,
yminorticks=true,
ylabel={RMSE($\hat{\mb{\varphi}}$)},
ymajorgrids,
yminorgrids,
axis background/.style={fill=white},
legend style={at={(0.03,0.03)},anchor=south west,legend cell align=left,align=left,draw=black},
height=\ph,
width=\pw,
xlabel style={yshift=-1mm},
ylabel near ticks
]
\addplot [color=\cGLC,solid,mark=\mGLC,mark options={solid}]
  table[row sep=crcr]{%
-10	1.22749293836992\\
-8	1.05489702025617\\
-6	0.726856883165392\\
-4	0.482719047646296\\
-2	0.363970723050297\\
0	0.310470540036058\\
2	0.253003178219903\\
4	0.215464388990843\\
6	0.185819331403465\\
8	0.16248373385038\\
10	0.139964541809402\\
12	0.122682802980647\\
14	0.106482266120066\\
16	0.0961518259777829\\
18	0.0861063008652121\\
20	0.0771582871445496\\
22	0.0706283376942376\\
24	0.0659486152118093\\
26	0.0619563238292041\\
28	0.0597889225752658\\
30	0.0569979296905475\\
};
\addlegendentry{GL-\mname{} \eqref{eq:smr_dual_gridless}};

\addplot [color=\cGBC,solid,mark=\mGBC,mark options={solid}]
  table[row sep=crcr]{%
-10	1.21758572453253\\
-8	1.0585173423724\\
-6	0.721537918752564\\
-4	0.48611914529395\\
-2	0.363977187920537\\
0	0.31089399285673\\
2	0.253260284061797\\
4	0.214502989275318\\
6	0.186250273944159\\
8	0.162809439100807\\
10	0.140630509713621\\
12	0.123423006227274\\
14	0.107506676598776\\
16	0.0978723397384218\\
18	0.0884066795345342\\
20	0.0804447279009728\\
22	0.0743935797545761\\
24	0.0716188198851384\\
26	0.0674389270339833\\
28	0.070254164218889\\
30	0.0637874460856308\\
};
\addlegendentry{\mname{} \eqref{eq:sdp1b}};

\addplot [color=\cRR,solid,mark=\mRR,mark options={solid}]
  table[row sep=crcr]{%
-10	3.41275907803159\\
-8	2.47399267365021\\
-6	1.74735937382186\\
-4	0.910482655947431\\
-2	0.608709139075908\\
0	0.467557940632244\\
2	0.397350157578282\\
4	0.375108529737277\\
6	0.359109710696278\\
8	0.354312360119856\\
10	0.349962767533216\\
12	0.347817125631041\\
14	0.346365109937375\\
16	0.345358442346804\\
18	0.344501221198319\\
20	0.344224037260704\\
22	0.344470545230528\\
24	0.344245656611603\\
26	0.344001007666797\\
28	0.34389072207532\\
30	0.343839088014222\\
};
\addlegendentry{Root-RARE \cite{Pesavento:Rare,Parvazi:PCA}};

\addplot [color=\cSR,solid,mark=\mSR,mark options={solid}]
  table[row sep=crcr]{%
-10	3.84496241140319\\
-8	2.34590568106844\\
-6	2.30893756788634\\
-4	0.926549168374893\\
-2	0.696076839118294\\
0	0.456753519105156\\
2	0.397689958641077\\
4	0.373087950709889\\
6	0.360009441365365\\
8	0.354944831694626\\
10	0.35034716889408\\
12	0.348169308209814\\
14	0.346738538257286\\
16	0.345797222182092\\
18	0.34485774836505\\
20	0.344641792447805\\
22	0.34491744842889\\
24	0.344759215540806\\
26	0.344426380479574\\
28	0.34435322930559\\
30	0.344269465928246\\
};
\addlegendentry{Spect. RARE \cite{Gershman:Rare,Parvazi:PCA}};

\addplot [color=black,solid]
  table[row sep=crcr]{%
-10	1.22424835038614\\
-8	0.869784014173624\\
-6	0.633077938493039\\
-4	0.471019467106333\\
-2	0.356964834157626\\
0	0.274456930648247\\
2	0.213274159808473\\
4	0.166973323180965\\
6	0.131388894031137\\
8	0.103735858295269\\
10	0.0820821910187327\\
12	0.0650399974844839\\
14	0.0515825697987074\\
16	0.040933049803069\\
18	0.0324939835037933\\
20	0.0258007097576016\\
22	0.0204891282614102\\
24	0.0162725342468978\\
26	0.0129244506971767\\
28	0.0102656131523651\\
30	0.00815394410799789\\
};
\addlegendentry{CRB \cite{Gershman:Rare}};

\end{axis}
\end{tikzpicture}%

%% file: comp_time_NSnp.tex
%
\begin{tikzpicture}

\begin{axis}[%
separate axis lines,
every outer x axis line/.append style={black},
every x tick label/.append style={font=\color{black}},
xmode=log,
xmin=1,
xmax=100,
xminorticks=true,
xlabel={Number of Snapshots $N$},
xmajorgrids,
xminorgrids,
every outer y axis line/.append style={black},
every y tick label/.append style={font=\color{black}},
ymode=log,
ymin=0.1,
ymax=1000,
yminorticks=true,
ylabel={Av. Computation Time in Secs},
ymajorgrids,
yminorgrids,
axis background/.style={fill=white},
legend style={at={(0.03,0.97)},anchor=north west,legend cell align=left,align=left,draw=black},
height=\ph,
width=\pw,
xlabel near ticks,
ylabel near ticks
]
\addplot [color=\cSR,solid,mark=\mSR,mark options={solid}]
  table[row sep=crcr]{%
1	3.1062548\\
2	3.47395732\\
3	3.902847\\
5	4.88795789\\
7	6.29146846\\
10	8.81900845\\
15	16.22008593\\
20	29.19677429\\
30	61.26895101\\
40	101.3367443\\
50	156.53710735\\
60	223.10187576\\
};
\addlegendentry{$\ell_{*,1}$ Mixed-Norm \eqref{eq:NucNormSdp2}};

\addplot [color=\cRR,solid,mark=\mRR,mark options={solid}]
  table[row sep=crcr]{%
1	2.51915757\\
2	2.96643931\\
3	3.0787832\\
5	3.27087986\\
7	3.51181203\\
10	4.02220898\\
15	5.18445715\\
20	7.85133886\\
30	24.35881956\\
40	109.62603382\\
50	200.66463252\\
60	473.79655435\\
};
\addlegendentry{\mname{} \eqref{eq:sdp1}};

\addplot [color=\cGBC,solid,mark=\mGBC,mark options={solid}]
  table[row sep=crcr]{%
1	3.91614038\\
2	3.54908903\\
3	3.53540006\\
5	3.63051409\\
7	3.73886932\\
10	3.70409693\\
15	3.71140219\\
20	3.70928612\\
30	3.68283313\\
40	3.84676887\\
50	3.77148288\\
60	3.70195718\\
};
\addlegendentry{\mname{} \eqref{eq:sdp1b}};

\addplot [color=\cGLC,solid,mark=\mGLC,mark options={solid}]
  table[row sep=crcr]{%
1	1.28662459\\
2	0.880664\\
3	0.90893727\\
5	0.93916716\\
7	0.81996504\\
10	0.65418712\\
15	0.65701158\\
20	0.65853767\\
30	0.65642274\\
40	0.65193745\\
50	0.65088358\\
60	0.64943909\\
};
\addlegendentry{GL-\mname{} \eqref{eq:smr_dual_gridless}};

\end{axis}
\end{tikzpicture}%

%% file: comp_time_NGrd.tex
%
\begin{tikzpicture}

\begin{axis}[%
separate axis lines,
every outer x axis line/.append style={black},
every x tick label/.append style={font=\color{black}},
xmode=log,
xmin=10,
xmax=1000,
xminorticks=true,
xlabel={Grid Size $K$},
xmajorgrids,
xminorgrids,
every outer y axis line/.append style={black},
every y tick label/.append style={font=\color{black}},
ymode=log,
ymin=0.1,
ymax=1000,
yminorticks=true,
ylabel={Av. Computation Time in Secs},
ymajorgrids,
yminorgrids,
axis background/.style={fill=white},
legend style={at={(0.03,0.97)},anchor=north west,legend cell align=left,align=left,draw=black},
height=\ph,
width=\pw,
xlabel near ticks,
ylabel near ticks
]
\addplot [color=\cSR,solid,mark=\mSR,mark options={solid}]
  table[row sep=crcr]{%
10	1.05961955\\
20	1.61689056\\
30	2.33127515\\
50	3.95617398\\
70	5.41857236\\
100	7.53112374\\
200	15.29891118\\
300	25.77914778\\
500	59.48592712\\
700	90.35917598\\
1000	118.69768459\\
};
\addlegendentry{$\ell_{*,1}$ Mixed-Norm \eqref{eq:NucNormSdp2}};

\addplot [color=\cRR,solid,mark=\mRR,mark options={solid}]
  table[row sep=crcr]{%
10	0.73072032\\
20	1.0018321\\
30	1.51019935\\
50	2.11462656\\
70	2.75196929\\
100	3.76384756\\
200	7.47611049\\
300	11.27989419\\
500	18.99842086\\
700	28.95516127\\
1000	45.24974564\\
};
\addlegendentry{\mname{} \eqref{eq:sdp1}};

\addplot [color=\cGBC,solid,mark=\mGBC,mark options={solid}]
  table[row sep=crcr]{%
10	0.78232405\\
20	1.11977708\\
30	1.42998761\\
50	2.04210883\\
70	2.63625095\\
100	3.64518573\\
200	7.36201249\\
300	11.34917348\\
500	20.49819631\\
700	31.51159387\\
1000	49.76819241\\
};
\addlegendentry{\mname{} \eqref{eq:sdp1b}};

\addplot [color=\cGLC,solid,mark=\mGLC,mark options={solid}]
  table[row sep=crcr]{%
10	0.64180402\\
20	0.64027997\\
30	0.64118219\\
50	0.64214521\\
70	0.64079396\\
100	0.63785932\\
200	0.64515588\\
300	0.71120161\\
500	0.64923726\\
700	0.73995632\\
1000	0.76130822\\
};
\addlegendentry{GL-\mname{} \eqref{eq:smr_dual_gridless}};

\end{axis}
\end{tikzpicture}%

%% file: ms.bbl
\begin{thebibliography}{10}
\providecommand{\url}[1]{#1}
\csname url@samestyle\endcsname
\providecommand{\newblock}{\relax}
\providecommand{\bibinfo}[2]{#2}
\providecommand{\BIBentrySTDinterwordspacing}{\spaceskip=0pt\relax}
\providecommand{\BIBentryALTinterwordstretchfactor}{4}
\providecommand{\BIBentryALTinterwordspacing}{\spaceskip=\fontdimen2\font plus
\BIBentryALTinterwordstretchfactor\fontdimen3\font minus
  \fontdimen4\font\relax}
\providecommand{\BIBforeignlanguage}[2]{{%
\expandafter\ifx\csname l@#1\endcsname\relax
\typeout{** WARNING: IEEEtran.bst: No hyphenation pattern has been}%
\typeout{** loaded for the language `#1'. Using the pattern for}%
\typeout{** the default language instead.}%
\else
\language=\csname l@#1\endcsname
\fi
#2}}
\providecommand{\BIBdecl}{\relax}
\BIBdecl

\bibitem{Krim:TwoDecades}
H.~Krim and M.~Viberg, ``Two decades of array signal processing research: the
  parametric approach,'' \emph{IEEE Signal Processing Magazine}, vol.~13,
  no.~4, pp. 67--94, Jul 1996.

\bibitem{575894}
B.~Porat and B.~Friedlander, ``Accuracy requirements in off-line array
  calibration,'' \emph{IEEE Transactions on Aerospace and Electronic Systems},
  vol.~33, no.~2, pp. 545--556, April 1997.

\bibitem{1165144}
Y.~Rockah and P.~Schultheiss, ``Array shape calibration using sources in
  unknown locations--part {I}: Far-field sources,'' \emph{IEEE Transactions on
  Acoustics, Speech and Signal Processing}, vol.~35, no.~3, pp. 286--299, Mar
  1987.

\bibitem{543678}
A.~J. Weiss and B.~Friedlander, ``Self-calibration in high-resolution array
  processing,'' in \emph{Advances in Spectrum Estimation and Array Processing},
  S.~Haykin, Ed.\hskip 1em plus 0.5em minus 0.4em\relax Prentice-Hall, 1991,
  vol.~II.

\bibitem{340783}
M.~Viberg and A.~Swindlehurst, ``A {Bayesian} approach to auto-calibration for
  parametric array signal processing,'' \emph{IEEE Transactions on Signal
  Processing}, vol.~42, no.~12, pp. 3495--3507, Dec 1994.

\bibitem{509886}
B.~C. Ng and C.~M.~S. See, ``Sensor-array calibration using a
  maximum-likelihood approach,'' \emph{IEEE Transactions on Antennas and
  Propagation}, vol.~44, no.~6, pp. 827--835, Jun 1996.

\bibitem{flanagan2001array}
B.~P. Flanagan and K.~L. Bell, ``Array self-calibration with large sensor
  position errors,'' \emph{Signal Processing}, vol.~81, no.~10, pp. 2201--2214,
  2001.

\bibitem{Gershman:Rare}
C.~See and A.~Gershman, ``Direction-of-arrival estimation in partly calibrated
  subarray-based sensor arrays,'' \emph{IEEE Transactions on Signal
  Processing}, vol.~52, no.~2, pp. 329--338, 2004.

\bibitem{wax1985decentralized}
M.~Wax and T.~Kailath, ``Decentralized processing in sensor arrays,''
  \emph{IEEE transactions on acoustics, speech, and signal processing},
  vol.~33, no.~5, pp. 1123--1129, 1985.

\bibitem{stoica1995decentralized}
P.~Stoica, A.~Nehorai, and T.~S{\"o}derstr{\"o}m, ``Decentralized array
  processing using the {MODE} algorithm,'' \emph{Circuits, Systems, and Signal
  Processing}, vol.~14, no.~1, pp. 17--38, 1995.

\bibitem{6854008}
W.~Suleiman and P.~Parvazi, ``Search-free decentralized direction-of-arrival
  estimation using common roots for non-coherent partly calibrated arrays,'' in
  \emph{Acoustics, Speech and Signal Processing (ICASSP), 2014 IEEE
  International Conference on}, May 2014, pp. 2292--2296.

\bibitem{1600024}
M.~F. Duarte, S.~Sarvotham, D.~Baron, M.~B. Wakin, and R.~G. Baraniuk,
  ``Distributed compressed sensing of jointly sparse signals,'' in
  \emph{Conference Record of the Thirty-Ninth Asilomar Conference onSignals,
  Systems and Computers, 2005.}, October 2005, pp. 1537--1541.

\bibitem{6880754}
Z.~Lu, R.~Ying, S.~Jiang, P.~Liu, and W.~Yu, ``Distributed compressed sensing
  off the grid,'' \emph{IEEE Signal Processing Letters}, vol.~22, no.~1, pp.
  105--109, Jan 2015.

\bibitem{960398}
A.~Swindlehurst, P.~Stoica, and M.~Jansson, ``Exploiting arrays with multiple
  invariances using {MUSIC} and {MODE},'' \emph{IEEE Transactions on Signal
  Processing}, vol.~49, no.~11, pp. 2511--2521, Nov 2001.

\bibitem{Parvazi:PCA}
P.~Parvazi, M.~Pesavento, and A.~Gershman, ``Direction-of-arrival estimation
  and array calibration for partly-calibrated arrays,'' in \emph{2011 IEEE
  International Conference on Acoustics, Speech and Signal Processing
  (ICASSP)}, 2011, pp. 2552--2555.

\bibitem{127959}
A.~Swindlehurst, B.~Ottersten, R.~Roy, and T.~Kailath, ``Multiple invariance
  {ESPRIT},'' \emph{IEEE Transactions on Signal Processing}, vol.~40, no.~4,
  pp. 867--881, Apr 1992.

\bibitem{6811813}
W.~Suleiman, M.~Pesavento, and A.~Zoubir, ``Decentralized direction finding
  using partly calibrated arrays,'' in \emph{2013 Proceedings of the 21st
  European Signal Processing Conference (EUSIPCO)}, Sept 2013, pp. 1--5.

\bibitem{6882325}
W.~Suleiman, P.~Parvazi, M.~Pesavento, and A.~Zoubir, ``Decentralized direction
  finding using {Lanczos} method,'' in \emph{Sensor Array and Multichannel
  Signal Processing Workshop (SAM), 2014 IEEE 8th}, June 2014, pp. 9--12.

\bibitem{Pesavento:Rare}
M.~Pesavento, A.~Gershman, and K.~M. Wong, ``Direction finding in partly
  calibrated sensor arrays composed of multiple subarrays,'' \emph{IEEE
  Transactions on Signal Processing}, vol.~50, no.~9, pp. 2103--2115, 2002.

\bibitem{steffens2017shiftinvariance}
C.~Steffens, W.~Suleiman, and M.~Sorg, A.~Pesavento, ``Gridless compressed
  sensing under shift-invariant sampling,'' in \emph{The 42nd IEEE
  International Conference on Acoustics, Speech and Signal Processing
  (ICASSP)}, March 2017.

\bibitem{6882328}
C.~Steffens, P.~Parvazi, and M.~Pesavento, ``Direction finding and array
  calibration based on sparse reconstruction in partly calibrated arrays,'' in
  \emph{Sensor Array and Multichannel Signal Processing Workshop (SAM), 2014
  IEEE 8th}, June 2014, pp. 21--24.

\bibitem{Malioutov:LassoDoa}
D.~Malioutov, M.~\c{C}etin, and A.~Willsky, ``A sparse signal reconstruction
  perspective for source localization with sensor arrays,'' \emph{IEEE
  Transactions on Signal Processing}, vol.~53, no.~8, pp. 3010--3022, 2005.

\bibitem{steffens2016compact}
C.~{Steffens}, M.~{Pesavento}, and M.~E. {Pfetsch}, ``A compact formulation for
  the $\ell_{2,1}$ mixed-norm minimization problem,'' \emph{ArXiv e-prints,
  arXiv:1606.07231v1}, Jun. 2016.

\bibitem{Tibshirani:Lasso}
R.~Tibshirani, ``Regression shrinkage and selection via the lasso,''
  \emph{Journal of the Royal Statistical Society. Series B (Methodological)},
  vol.~58, pp. 267--288, 1996.

\bibitem{Chen98atomicdecomposition}
S.~S. Chen, D.~L. Donoho, and M.~A. Saunders, ``Atomic decomposition by basis
  pursuit,'' \emph{SIAM Journal On Scientific Computing}, vol.~20, pp. 33--61,
  1998.

\bibitem{candes2012a}
E.~J. Cand{\`e}s and C.~Fernandez-Granda, ``Towards a mathematical theory of
  super-resolution,'' \emph{Communications on Pure and Applied Mathematics},
  vol.~67, no.~6, pp. 906--956, 2014.

\bibitem{candes2012b}
------, ``Super-resolution from noisy data,'' \emph{Journal of Fourier Analysis
  and Applications}, vol.~19, no.~6, pp. 1229--1254, 2013.

\bibitem{tang2013}
G.~Tang, B.~Bhaskar, P.~Shah, and B.~Recht, ``Compressed sensing off the
  grid,'' \emph{IEEE Transactions on Information Theory}, vol.~59, no.~11, pp.
  7465--7490, Nov 2013.

\bibitem{Boyd:RankMinimization}
M.~Fazel, H.~Hindi, and S.~Boyd, ``A rank minimization heuristic with
  application to minimum order system approximation,'' in \emph{Proceedings of
  the American Control Conference}, vol.~6, 2001, pp. 4734--4739.

\bibitem{Recht2010}
B.~Recht, M.~Fazel, and P.~A. Parrilo, ``Guaranteed minimum-rank solutions of
  linear matrix equations via nuclear norm minimization,'' \emph{SIAM Review},
  vol.~52, no.~3, pp. 471--501, 2010.

\bibitem{yuan2006grouplasso}
Y.~L. Ming~Yuan, ``Model selection and estimation in regression with grouped
  variables,'' \emph{Journal of the Royal Statistical Society. Series B
  (Statistical Methodology)}, vol.~68, no.~1, pp. 49--67, 2006.

\bibitem{kowalski2009mixednorm}
M.~Kowalski, ``Sparse regression using mixed norms,'' \emph{Applied and
  Computational Harmonic Analysis}, vol.~27, no.~3, pp. 303 -- 324, 2009.

\bibitem{steffens2016mimo}
C.~Steffens, Y.~Yang, and M.~Pesavento, ``Multidimensional sparse recovery for
  {MIMO} channel parameter estimation,'' in \emph{Proceedings of the 2016
  European Signal Processing Conference}, Budapest, Hungary, September 2016.

\bibitem{steffens2016noncircular}
J.~Steinwandt, C.~Steffens, M.~Pesavento, and M.~Haardt, ``Sparsity-aware
  direction finding for strictly non-circular sources based on rank
  minimization,'' July 2016.

\bibitem{S98guide}
J.~Sturm, ``Using {SeDuMi} 1.02, a {MATLAB} toolbox for optimization over
  symmetric cones,'' \emph{Optimization Methods and Software}, vol. 11--12, pp.
  625--653, 1999.

\bibitem{vandenberghe1996semidefinite}
L.~Vandenberghe and S.~Boyd, ``Semidefinite programming,'' \emph{SIAM review},
  vol.~38, no.~1, pp. 49--95, 1996.

\bibitem{Dumitrescu:1086500}
B.~Dumitrescu, \emph{Positive Trigonometric Polynomials and Signal Processing
  Applications}.\hskip 1em plus 0.5em minus 0.4em\relax Berlin: Springer, 2007.

\bibitem{pesavento2005fast}
M.~Pesavento, ``Fast algorithms for multidimensional harmonic retrieval,''
  Ph.D. dissertation, Ruhr-Universit{\"a}t Bochum, 2005.

\bibitem{barabell1983rootMusic}
A.~Barabell, ``Improving the resolution performance of eigenstructure-based
  direction-finding algorithms,'' in \emph{Acoustics, Speech, and Signal
  Processing, IEEE International Conference on ICASSP '83.}, vol.~8, Apr 1983,
  pp. 336--339.

\bibitem{pesavento2000unitaryRootMusic}
M.~Pesavento, A.~B. Gershman, and M.~Haardt, ``Unitary root-{MUSIC} with a
  real-valued eigendecomposition: a theoretical and experimental performance
  study,'' \emph{IEEE Transactions on Signal Processing}, vol.~48, no.~5, pp.
  1306--1314, May 2000.

\bibitem{Abramovich:EL}
E.~Northardt, I.~Bilik, and Y.~Abramovich, ``Spatial compressive sensing for
  direction-of-arrival estimation with bias mitigation via expected
  likelihood,'' \emph{IEEE Transactions on Signal Processing}, vol.~61, no.~5,
  pp. 1183--1195, 2013.

\bibitem{grant2008}
M.~Grant and S.~Boyd, ``Graph implementations for nonsmooth convex programs,''
  in \emph{Recent Advances in Learning and Control}, ser. Lecture Notes in
  Control and Information Sciences, V.~Blondel, S.~Boyd, and H.~Kimura,
  Eds.\hskip 1em plus 0.5em minus 0.4em\relax Springer-Verlag Limited, 2008,
  pp. 95--110.

\bibitem{grant2014}
------, ``{CVX}: Matlab software for disciplined convex programming, version
  2.1,'' \url{http://cvxr.com/cvx}, Mar. 2014.

\bibitem{Srebro2005}
N.~Srebro and A.~Shraibman, ``Rank, trace-norm and max-norm,'' in
  \emph{Proceedings of the 18th Annual Conference on Learning Theory}.\hskip
  1em plus 0.5em minus 0.4em\relax Springer-Verlag, 2005, pp. 545--560.

\bibitem{searle1982}
S.~Searle, \emph{Matrix algebra useful for statistics}, ser. Wiley series in
  probability and mathematical statistics: Applied probability and
  statistics.\hskip 1em plus 0.5em minus 0.4em\relax Wiley, 1982.

\end{thebibliography}
